\documentclass[5p,twocolumn]{elsarticle} 

\usepackage{graphicx}
\graphicspath{{./Figures/}}

\usepackage{amsmath}
\usepackage{amssymb}  
\usepackage{amsfonts} 
\usepackage{amsthm}   
\usepackage{enumitem} 
\usepackage{array}
\usepackage[hyphens]{url}
\usepackage{tabu} 

\usepackage{tikz}
\usetikzlibrary{arrows}
\usetikzlibrary{shapes}
\usetikzlibrary{automata}
\usetikzlibrary{positioning}
\usetikzlibrary{decorations.markings}

\journal{Journal of Control Engineering Practice}

\newtheorem{theorem}{Theorem}
\newtheorem{lemma}{Lemma}

\newtheorem{definition}{Definition}

\newcommand{\CNMS}{\textbf{CNMS} }
\newcommand{\RCNMS}{\textbf{RCNMS}}

\definecolor{lbl}{RGB}{55, 124, 223}

\biboptions{authoryear}

\newcommand{\hl}[1]{\textcolor{black}{#1}}
\newcommand{\hll}[1]{\textcolor{black}{#1}}
\newcommand{\hlll}[1]{\textcolor{black}{#1}}
\newcommand{\hllll}[1]{\textcolor{black}{#1}}

\newcommand{\hlr}[1]{\textcolor{black}{#1}}

\newcommand{\hlrr}[1]{\textcolor{black}{#1}}

\begin{document}

\begin{frontmatter}

\title{Model Properties for \hl{Efficient Synthesis of}\\ Nonblocking Modular Supervisors}


\author[tue]{Martijn~Goorden\corref{mycorrespondingauthor}}
\cortext[mycorrespondingauthor]{Corresponding author}
\ead{m.a.goorden@tue.nl}

\author[tue]{Joanna~van~de~Mortel-Fronczak}
\ead{j.m.v.d.mortel@tue.nl}

\author[tue]{Michel~Reniers}
\ead{m.a.reniers@tue.nl}

\author[chalmers]{Martin~Fabian}
\ead{fabian@chalmers.se}

\author[vu,tue]{Wan~Fokkink}
\ead{w.j.fokkink@vu.nl}

\author[tue]{Jacobus~Rooda}
\ead{j.e.rooda@tue.nl}

\address[tue]{Department of Mechanical Engineering, Eindhoven University of Technology, Eindhoven, The Netherlands}
\address[chalmers]{Department of Electrical Engineering, Chalmers University of Technology, G\"{o}teborg, Sweden}
\address[vu]{Department of Computer Science, Vrije Universiteit Amsterdam, Amsterdam, The Netherlands}

\begin{abstract}
Supervisory control theory provides means to synthesize supervisors for \hlll{systems with discrete-event behavior} from models of the uncontrolled plant and of the control requirements. 
\hlr{The applicability of supervisory control theory often fails due to a lack of scalability of the algorithms.}
\hlrr{This paper proposes} \hlr{a format for the requirements} and a method \hlr{to ensure that the crucial properties of controllability and nonblockingness directly hold, thus avoiding the most computationally expensive parts of synthesis.} The method consists of creating a control problem dependency graph and verifying whether it is acyclic. \hlr{Vertices of the graph are modular plant components, and edges are derived from the requirements.} In case of a cyclic graph, potential blocking issues can be localized, so that the original control problem can be reduced to only synthesizing supervisors for smaller partial control problems. The \hlr{strength of the} method is illustrated \hlr{on two} case stud\hlr{ies:} a production line and a roadway tunnel. 
\end{abstract}


\end{frontmatter}


\section{Introduction}\label{sec:introduction}

The design of supervisors for \hlll{systems with discrete-event behavior} has become a challenge as \hlll{they comprise growing numbers of} components to control and \hlll{of} functions to fulfill, at the same time \hlll{being subject to} market demands requiring verified safety, decreased costs, and decreased time-to-market. Model-based systems engineering methodologies can help in overcoming these difficulties.

The supervisory control theory of Ramadge-Wonham~\cite{ramadge_supervisory_1987,ramadge_control_1989} provides means to synthesize supervisors from a model of the uncontrolled plant and a model of the control requirements. \hlll{This synthesis step,} \hlr{in which controllable events that lead to violations of requirements are disabled,} guarantees by construction that the closed-loop behavior of the supervisor and the plant adheres to all requirements, is controllable, nonblocking, and maximally permissive. \hlr{In this context, controllable means that only controllable events are disabled, and nonblocking that special (marked) system states remain reachable. That the supervisor is maximally permissive means that it restricts the controlled behavior as little as absolutely necessary to fulfill the other properties.}



\hlll{It has been shown in~\cite{gohari_complexity_2000} that synthesis is NP-hard, which means that any algorithm computing a supervisor will in the worst case have exponential time and memory complexity. Thus, supervisor synthesis is a tough problem.
However, as was also noted by~\cite{gohari_complexity_2000}, by observing real-world problems more closely one could discover instances of supervisory control problems that are computationally easier. No suggestions were included \hllll{in that paper} of what these instances might be or how to find them, though.}

Analyzing \hlll{a number of models of industrial-size applications, including the recently published~\citep{reijnen_supervisory_2017,reijnen_application_2018,reijnen_supervisory_2019,moormann_design_2020}}, one discovers that the synthesized supervisors do not impose \hlll{any} additional restrictions on the system, i.e., the provided set of requirement models is \hlll{already} sufficient to control the plant such that the closed-loop behavior is controllable, nonblocking, and maximally permissive. \hlr{That is, requirements do not disable uncontrollable events and do not violate nonblocking.} Therefore, time and computing resources could have been saved, as \hlll{the synthesis step} turned out to be unnecessary for these cases
. 

\hlll{When developing the kind of large infrastructural systems \hllll{considered in e.g.}~\cite{reijnen_supervisory_2017,reijnen_supervisory_2019,moormann_design_2020}, due to their safety-criticality, engineers \hllll{tend to} follow a strictly structured development process based on failure mode analysis~\citep{stamatis_failure_2003}. This process includes} \hlr{decomposing the system into components, resulting in a modular plant where no events are shared between components. Furthermore, }  \hlll{requirement specifications \hlr{are formulated} in a specific way, where an actuator action is guarded by requirements on the states of other components, typically sensors but also other actuators. So the requirements models are formulated as state-based expressions~\citep{ma_nonblocking_2005,markovski_coordination_2010}
.}
\hllll{As shown in this paper,} \hlll{this specific formulation is beneficial for supervisor synthesis}\hlr{:} \hlll{it can be determined beforehand that synthesis will not impose additional restrictions on the system. Thus, engineers of safety-critical systems working with the specific formulation obtain a powerful tool to gain confidence in the obtained supervisor.
This is a crucial aspect for the supervisory control theory, which after 30 years of academic research with impressive results, still \hllll{has not gained} wide-spread industrial acceptance.}

\hlll{This paper \hlrr{takes} a different approach compared to an often used methodology} in supervisory control synthesis, where particular structures of systems are used to ease synthesis and \hll{which are} applicable to any given discrete-event system model. Examples of such methods include local modular synthesis~\citep{queiroz_modular_2000-1}, incremental synthesis~\citep{brandin_incremental_2004}, compositional synthesis~\citep{mohajerani_framework_2014}, and coordination control~\citep{komenda_coordination_2014}. \hl{Experimenting with applying several of these synthesis methods directly on the full models of~\cite{reijnen_supervisory_2017,moormann_design_2020} shows that applying the wrong algorithm is fatal in the sense of running out of memory. Thus, knowing beforehand that synthesis is not necessary will save computational time and effort.}

\hlll{This paper brings three contributions, all aimed at easing the synthesis effort necessary for large-scale industrial supervisory control problems. First, it is shown that if \hlr{the plant is formulated modularly and} the requirements specifications are expressed in a way standard to some engineering practices,
and in addition follow some simple and easily checked rules called \CNMS (laid out in detail in Section~\ref{sect:properties}), then \emph{the synthesis step is altogether unnecessary}: the plant and the requirement models already form a controllable, nonblocking, and maximally permissive supervisor.
Knowing this beforehand is a huge benefit, as for such large-scale models state-of-the-art synthesis algorithms may take a long time, or may even not be able to synthesize a supervisor at all due to memory or time constraints.
}

\hlll{However, the \CNMS rules are rather conservative. Examples are known that do not fulfill the rules, but for which synthesis is still not necessary. Thus, the second contribution relaxes \CNMS into \RCNMS. Dependencies between the \hlr{modular} plant models based on the requirements \hlr{are captured} \hlr{with a \emph{dependency graph}} (which does not need enumerating the actual state-space), \hlr{where vertices relate to the plant models and the edges to the requirement models}. \hlr{By analyzing a dependency graph,} it can be determined \hlr{for \RCNMS\ rules} whether synthesis will be necessary or not. If \hlr{there are} no cycles, then the synthesis step is unnecessary: the plant and the requirement models already form a controllable, nonblocking, and maximally permissive supervisor. Note that, instead of using the more common suggestions found in the SCT literature to analyze the dependencies between plant models (e.g., shared events in~\cite{flordal_compositional_2009,komenda_multilevel_2013}), the dependencies within the combined set of plant models and the requirement models \hlrr{is analyzed}, as also suggested in~\cite{feng_computationally_2006,goorden_structuring_2020}}.

\hlll{In general, the dependency graph has cyclic parts, \hlr{though}. The third contribution of this paper shows that synthesis can then be sectionalized, so as to be performed only for those plant and requirement models that are involved in the strongly connected components of those cycles. The other parts need no synthesis for the same reason as above. \hlr{This results in modular supervisors, where a collection of supervisors controls the plant together.} This \hlr{contribution} reduces the synthesis effort significantly, such that supervisors could now be synthesized for control problems where state-of-the-art synthesis algorithms fail (as demonstrated in Section~\ref{sect:tunnel}).
}

\hll{The proposed method is most effective for plant models that are loosely coupled, often the result of using the input/output perspective}~\hlll{\citep{balemi_control_1992}.} \hll{Several case studies with real-life applications, such as~\cite{reijnen_supervisory_2017,reijnen_application_2018,reijnen_supervisory_2019,moormann_design_2020}, have loosely coupled plant models,} \hlll{and experiments show that these models benefit greatly from the described approach.}
\hlr{Even} \hlll{though some well-known examples~\citep{sanden_modular_2015,su_aggregative_2010,wonham_supervisory_2018}} \hllll{are hard to fit} \hlll{in the presented framework,}
\hl{the key point is that \emph{if} a system is modeled with the specific formulation described in this paper, \hlll{\emph{then}} a reduction in synthesis effort can be achieved.}

This paper builds upon preliminary results published in~\cite{goorden_no_2019}. While the \CNMS model properties proposed in that paper capture the essence of some models of industrial applications, this paper \hlrr{generalizes} \hlll{the approach by} providing relaxed conditions \hlr{to determine that} \hlll{synthesis is not necessary}. 

\hlr{Related work is}~\cite{feng_computationally_2006}, where \hl{inspiration is taken from systems with shared resources, such as flexible manufacturing systems. In~\cite{feng_computationally_2006},} control-flow nets are introduced to analyze dependencies in the system and subsequently abstract away those parts of the system that will not contribute to a potential blocking issue. Control-flow nets are defined for shuffle systems with server and buffer specifications, which limits \hlr{their} applicability. In \hlrr{this paper}, \hl{a notion similar to a shuffle system for the plant} \hlr{is used}, while the specifications \hll{are \hlr{in terms of} state-based \hlr{expressions}, see~\cite{ma_nonblocking_2005,markovski_coordination_2010}}. Nevertheless, both works \hlr{complement each other as they identify} different classes of discrete-event systems for which synthesis is easy.


The structure of this paper is as follows. In Section~\ref{sect:preliminaries}, the preliminaries are provided. The \CNMS properties as proposed in previous work~\citep{goorden_no_2019} are presented in Section~\ref{sect:properties}, \hlll{and it is shown that \hlr{for models satisfying these} properties, synthesis is unnecessary}. Section~\ref{sect:graphanalysis} introduces the dependency graph used to analyze the control problems. In Section~\ref{sect:acyclicgraphs}, the result is established that synthesis \hlll{is unnecessary} when the dependency graph is acyclic. Section~\ref{sect:cyclicgraphs} extends the analysis to cyclic dependency graphs to reduce the original control problem into a set of smaller control problems. Sections~\ref{sect:casestudy} and~\ref{sect:tunnel} provide two case studies, related to a production line and to a roadway tunnel, to demonstrate the proposed analysis method. Section~\ref{sect:conclusion} concludes the paper.

\section{Preliminaries}\label{sect:preliminaries}
This section provides a brief summary of concepts related to automata, supervisory control theory, and directed graphs relevant for this paper. The concepts related to automata and supervisory control theory are taken from~\cite{cassandras_introduction_2008,wonham_supervisory_2018}. The concepts related to directed graphs are taken from~\cite{diestel_graph_2017}.

\subsection{Automata}
An automaton is a five-tuple $G=(Q, \Sigma,\delta, q_0,Q_m)$, where $Q$ is the (finite) state set, $\Sigma$ is the alphabet of events, $\delta:Q\times \Sigma \rightarrow Q$ the partial transition function, $q_0 \in Q$ the initial state, and $Q_m \subseteq Q$ the set of marked states.
The alphabet $\Sigma = \Sigma_c \cup \Sigma_u$ is partitioned into two disjoint sets containing the controllable events ($\Sigma_c$) and the uncontrollable events ($\Sigma_u$), and $\Sigma^*$ is the set of all finite strings of events in $\Sigma$, including empty string $\varepsilon$.

\hlrr{The notation} $\delta(q,\sigma)!$ \hlrr{denotes} that there exists a transition from state $q\in Q$ labeled with event $\sigma$, i.e., $\delta(q,\sigma)$ is defined. The transition function can be extended in a natural way to strings as \hlrr{$\delta(q,\varepsilon) = q$ for the empty string $\varepsilon$,}  $\delta(q,s\sigma) = \delta(\delta(q,s),\sigma)$ where $s\in\Sigma^*$, $\sigma\in\Sigma$, and $\delta(q,s\sigma)!$ if $\delta(q,s)!\wedge\delta(\delta(q,s),\sigma)!$. The language generated by the automaton $G$ is $\mathcal{L}(G) = \{s\in\Sigma^*\ |\ \delta(q_0,s)! \}$ and the language marked by the automaton $G$ is $\mathcal{L}_m(G)=\{s\in\hlr{\mathcal{L}(G)}\ |\ \delta(q_0,s)\in Q_m\}$.

A path $p$ of an automaton is defined as a sequence of alternating states and events, i.e., $q_1\sigma_1q_2\sigma_2\ldots \sigma_{n-1} q_n\sigma_nq_{n+1}$ \hlr{s.t.} $\forall i \in [1, n], \delta(q_i,\sigma_i)=q_{i+1}$. A path can also be written in the infix notation $q_1 \xrightarrow{\sigma_1} q_2 \xrightarrow{\sigma_2} \ldots \xrightarrow{\sigma_{n-1}} q_n \xrightarrow{\sigma_n} q_{n+1}$.

A state $q$ of an automaton is called \emph{reachable} if there is a string $s\in\Sigma^*$ with $\delta(q_0,s)!$ and $\delta(q_0,s) = q$. The automaton $G$ is called \emph{reachable} if every state $q\in Q$ is reachable. A state $q$ is \emph{coreachable} if there is a string $s\in\Sigma^*$ with $\delta(q,s)!$ and $\delta(q,s)\in Q_m$. The automaton $G$ is called \emph{coreachable} if every state $q\in Q$ is coreachable. An automaton is called \emph{nonblocking} if every reachable state is coreachable. An automaton is called \emph{trim} if it is reachable and coreachable. Notice that a trim automaton is nonblocking, but a nonblocking automaton may not be trim, since it may have unreachable states.

An automaton is called \emph{strongly connected} if from every state all other states can be reached, i.e., $\hlr{\forall} q_1,q_2 \in Q, \hlr{\exists} s\in\Sigma^*$ \hlr{s.t.} $\delta(q_1,s) = q_2$, see~\cite{ito_representation_1978}.

Two automata can be combined by synchronous composition.
\begin{definition}\label{def:synchronouscomposition}
Let $G_1 = (Q_1,\Sigma_1,\delta_1,q_{0,1},Q_{m,1})$, $G_2 = (Q_2,\Sigma_2,\delta_2,q_{0,2},Q_{m,2})$ be two automata. The synchronous composition of $G_1$ and $G_2$ is defined as
\begin{align*}
G_1\parallel G_2 = (&Q_1\times Q_2,\Sigma_1\cup\Sigma_2,\delta_{1\parallel 2},(q_{0,1},q_{0,2}),\\
 &Q_{m,1}\times Q_{m,2})
\end{align*}
where
\begin{align*}
  \delta_{1\parallel 2}&((x_1,x_2),\sigma) = \\
  &
  \begin{cases}
    (\delta_1(x_1,\sigma),\delta_2(x_2,\sigma)) &\text{if } \sigma\in\Sigma_1\cap\Sigma_2, \delta_1(x_1,\sigma)!,\\ &\text{and } \delta_2(x_2,\sigma)! \\
    (\delta_1(x_1,\sigma),x_2) &\text{if } \sigma\in\Sigma_1\setminus\Sigma_2 \text{ and } \delta_1(x_1,\sigma)! \\
    (x_1,\delta_2(x_2,\sigma)) &\text{if } \sigma\in\Sigma_2\setminus\Sigma_1 \text{ and } \delta_2(x_2,\sigma)!\\
    \text{undefined} &\text{otherwise.}
  \end{cases}
\end{align*}
\end{definition}
Synchronous composition is associative and commutative up to reordering of the state components in the composed state set.

A composed system $\mathcal{G}$ is a collection of automata, i.e., $\mathcal{G} = \{G_1,\ldots,G_m\}$. The synchronous composition of a composed system $\mathcal{G}$, denoted by $\parallel \mathcal{G}$, is defined as $\parallel \mathcal{G} = G_1 \parallel \ldots \parallel G_m$, and the synchronous composition of two composed systems $\mathcal{G}_1 \parallel \mathcal{G}_2$ is defined as $\parallel (\mathcal{G}_1 \cup \mathcal{G}_2)$. A composed system $\mathcal{G} = \{G_1,\ldots,G_m\}$ is called a \emph{product system} if the alphabets of the automata are pairwise disjoint, i.e., $\hlr{\forall} i,j\in [1,m], i\neq j,\Sigma_i\cap\Sigma_j=\emptyset$~\citep{ramadge_control_1989}.

Finally, let $G$ and $K$ be two automata with the same alphabet $\Sigma$. $K$ is said to be \emph{controllable} with respect to $G$ if, for every string $s\in\Sigma^*$ and $u\in\Sigma_u$ such that $\delta_K(q_{0,K},s)!$ and $\delta_G(q_{0,G},su)!$, it holds that $\delta_K(q_{0,K},su)!$.

\subsection{Supervisory control theory}\label{subsect:SCT}
The objective of supervisory control theory~\citep{ramadge_supervisory_1987,ramadge_control_1989,cassandras_introduction_2008,wonham_supervisory_2018} is to design an automaton called a \emph{supervisor} which function is to dynamically disable controllable events so that the closed-loop system of the plant and the supervisor obeys some specified behavior. More formally, given a plant model $P$ and requirement model $R$, the goal is to synthesize supervisor $S$ that adheres to the following control objectives.
\begin{itemize}
    \item \emph{Safety}: all possible behavior of the closed-loop system $P\parallel S$ should always satisfy the imposed requirements, i.e., $\mathcal{L}(P\parallel S) \subseteq \mathcal{L}(P\parallel R)$.
    \item \emph{Controllability}: uncontrollable events may never be disabled by the supervisor, i.e., $P\parallel S$ is controllable with respect to $P$.
    \item \emph{Nonblockingness}: the closed-loop system should be able to reach a marked state from every reachable state, i.e., $P\parallel S$ is nonblocking.
    \item \emph{Maximal permissiveness}: the supervisor does not restrict more behavior than strictly necessary to enforce safety, controllability, and nonblockingness, i.e., for all other supervisors $S'$ it holds that $\mathcal{L}(P\parallel S')\subseteq\mathcal{L}(P\parallel S)$.
\end{itemize}

\emph{Monolithic supervisory control synthesis} results in a single supervisor $S$ from a single plant model and a single requirement model~\citep{ramadge_supervisory_1987}. There may exist multiple automata representations of the safe, controllable, nonblocking, and maximally permissive supervisor. Without loss of generality it is assumed that $S = P\parallel S$. When the plant model and the requirement model are given as a composed system $\mathcal{P}$ and $\mathcal{R}$, respectively, the monolithic plant model $P$ and requirement model $R$ are obtained by performing the synchronous composition of the models in the respective composed system.

For the purpose of supervisor synthesis, requirements can be modeled with automata and \emph{state-based expressions}~\citep{ma_nonblocking_2005,markovski_coordination_2010}. The latter is useful in practice, as \hl{some control} engineers tend to formulate requirements based on states of the plant. To refer to states of the plant, the notation $P.q$ \hlrr{is introduced} that refers to state $q$ of plant $P$. State references can be combined with the Boolean literals $\textbf{T}$ and \textbf{F} and logic connectives to create \emph{predicates}.

In this paper, state-event invariant expressions are considered. A \emph{state-event invariant expression} formulates conditions on the enablement of an event based on states of the plant, i.e., the condition should evaluate to true for the event to be enabled. A state-event invariant expression is of the form $\sigma \textbf{ needs } C$ where $\sigma$ is an event and $C$ a predicate stating the condition. \hl{In general,} event $\sigma$ can be a controllable or an uncontrollable event. Let $R$ be a state-event invariant expression, then $\mathit{event}(R)$ returns the event used in $R$ and $\mathit{cond}(R)$ returns the condition predicate. The synchronous composition of a plant $P$ with a state-event invariant expression $R$, denoted with $P\parallel R$, is defined by altering the transition function $\delta$.
\begin{definition}\label{def:syncompreq}
Let $P=(Q, \Sigma,\delta, q_0,Q_m)$ and $R = \mu \textbf{ needs } C$. Then the synchronous composition of $P$ and $R$ is defined as
\begin{equation*}
P\parallel R = (Q, \Sigma,\delta', q_0,Q_m)
\end{equation*}
where $\delta'(q,\sigma) = \delta(q,\sigma)$, unless $\sigma=\mu$ and $C_{| P.q} = \textbf{F}$, where $C_{| P.q}$ indicates that all state references $P.q$ in $C$ are substituted by $\textbf{T}$ and all state references $P.r, r\in Q,r\neq q$ in $C$ replaced by $\textbf{F}$. In the latter case $\delta'(q,\sigma)$ is undefined.
\end{definition}

\begin{figure}
	\begin{center}
		\begin{tikzpicture}[->,>=stealth',shorten >=1pt,main node/.style={circle,draw,font=\small}, node distance=1.5cm]
			\node[main node, label={below:$q_1$}, initial, initial text={$P_1$},initial where=left, accepting] (l1) {};
			\node[main node, label={below:$q_2$}]  (l2) [right = of l1] {};
			\node[main node, label={below:$q_3$}, accepting] (l3) [right = of l2] {};
			
			\path[every node/.style={font=\sffamily\small}, bend left=15]
			(l1) edge node [above] {$a$} (l2)
			(l2) edge node [below] {$b$} (l1)
			(l2) edge node [above] {$a$} (l3)
			(l3) edge node [below] {$b$} (l2)
			;
			
			\node (r1) [right = .5cm of l3] {$R_1: b \textbf{ needs } P_1.q_3$};
			
			\node[main node, label={below:$q_1$}, initial, initial text={$P_1\parallel R_1$},initial where=left, accepting] (l11) [below = of l2] {};
			\node[main node, label={below:$q_2$}]  (l21) [right = of l11] {};
			\node[main node, label={below:$q_3$}, accepting] (l31) [right = of l21] {};
			
			\path[every node/.style={font=\sffamily\small}, bend left=15]
			(l11) edge node [above] {$a$} (l21)
			(l21) edge node [above] {$a$} (l31)
			(l31) edge node [below] {$b$} (l21)
			;
			
		\end{tikzpicture}
	\end{center}
	\caption{An example of the synchronous composition of an automaton and a state-event invariant expression. In this and subsequent figures, (marked) locations are depicted with (concentric) circles, the initial location with an incoming arrow, and transitions with labeled edges.}
	\label{fig:exampleStateEvent}
\end{figure}
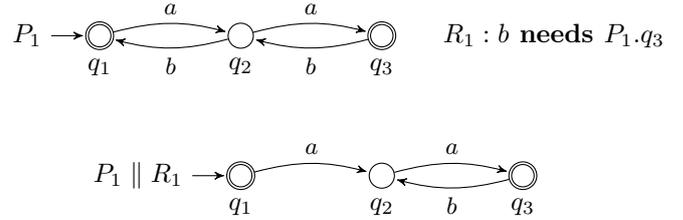

An example to illustrate the synchronous composition between an automaton and a state-event invariant expression is provided in Figure~\ref{fig:exampleStateEvent}.
This definition can be easily extended to a set of state-event invariant expressions $\mathcal{R} = \{R_1,\ldots, R_n\}$.

Given a composed system representation of the plant $\mathcal{P} = \{P_1,\ldots,P_m\}$ and a collection of requirements $\mathcal{R}=\{R_1,\ldots,R_n\}$, the tuple $(\mathcal{P},\mathcal{R})$ \hlrr{is defined} as the \emph{control problem} for which \hlrr{a supervisor needs to be synthesized}. The following assumptions \hlrr{are made} about this control problem:
\begin{itemize}
    \item $\mathcal{P} \neq \emptyset$, while $\mathcal{R}$ can be the empty set.
    \item For all $P\in \mathcal{P}$, it holds that $P$ is an automaton where $Q_{P}$ and $\Sigma_{P}$ are nonempty.
    \item For all $R\in \mathcal{R}$, it holds that
    \begin{itemize}
        \item if $R$ is an automaton, then $Q_{R}$ and $\Sigma_{R}$ are nonempty, and $\Sigma_{R}\subseteq \Sigma_{\mathcal{P}}$ where $\Sigma_{\mathcal{P}} = \bigcup_{P\in\mathcal{P}} \Sigma_{P}$,
        \item if $R$ is a state-event invariant expression, then $\mathit{event}(R) \in \Sigma_{\mathcal{P}}$, and for each state reference $P.q$ in $\mathit{cond}(R)$ it holds that $P\in \mathcal{P}$ and $q\in Q_{P}$.
    \end{itemize}
\end{itemize}

\emph{Modular supervisory control synthesis} uses the fact that often the desired behavior is specified with a collection of requirements $\mathcal{R}$~\citep{wonham_modular_1988}. Instead of first transforming the collection of requirements into a single requirement, as monolithic synthesis does, modular synthesis calculates for each requirement a supervisor based on the plant model. In other words, given a control problem $(\mathcal{P},\mathcal{R})$ with $\mathcal{R}=\{R_1,\ldots,R_n\}$, modular synthesis solves $n$ control problems $(\mathcal{P},\{R_1\}),\ldots, (\mathcal{P},\{R_n\})$. Each control problem $(\mathcal{P},\{R_i\})$ for $i\in[1,n]$ results in a safe, controllable, nonblocking, and maximally permissive supervisor $S_i$. The collection of supervisors $\mathcal{S}=\{S_1,\ldots,S_n\}$ can be conflicting, i.e., $S_1\parallel\ldots\parallel S_n$ can be blocking. A nonconflicting check can verify whether $\mathcal{S}$ is nonconflicting~\citep{mohajerani_framework_2016}. In the case that $\mathcal{S}$ is nonconflicting, $\mathcal{S}$ is also safe, controllable, nonblocking, and maximally permissive for the original control problem $(\mathcal{P},\mathcal{R})$. In the case that $\mathcal{S}$ is conflicting, an additional coordinator $C$ can be synthesized such that $\mathcal{S}\cup\{C\}$ is safe, controllable, nonblocking, and maximally permissive for the original control problem $(\mathcal{P},\mathcal{R})$~\citep{su_synthesize_2009}.

\subsection{Directed graphs}
Definitions and notations of directed graphs are taken from~\cite{diestel_graph_2017}. A \emph{directed graph} is a tuple $(V,E)$ of sets of vertices $V$ (or nodes) and edges $E$ (or arcs), together with two functions $\mbox{init}: E\rightarrow V$ and $\mbox{ter}:E\rightarrow V$. The function $\mbox{init}$ assigns to each edge $e$ an initial vertex $\mbox{init}(e)$ and the function $\mbox{ter}$ assigns to each edge $e$ a terminal vertex $\mbox{ter}(e)$. An edge $e$ is said to be directed from vertex $\mbox{init}(e)$ to vertex $\mbox{ter}(e)$. If $\mbox{init}(e)=\mbox{ter}(e)$, the edge $e$ is called a loop. A directed graph is called \emph{\hlll{self}-loop free} if no edge is a loop. A directed graph $G'=(V',E')$ is a subgraph of $G=(V,E)$, \hlr{denoted} by $G'\subseteq G$, if $V'\subseteq V$ and $E'\subseteq E$.

A \emph{path} in directed graph $G=(V,E)$ is a sequence of its vertices $p=x_0x_1\ldots x_k, k \geq 0$ such that for each step $i\in[0,k-1]$ there exists an edge $e_i\in E$ with $\mbox{init}(e_i) = x_i$ and $\mbox{ter}(e_i) = x_{i+1}$. The path $p$ is also called a path from $x_0$ to $x_k$. Two paths $p_1=x_0\ldots x_k$ and $p_2=y_0\ldots y_l$ can be concatenated into path $p_1p_2=x_0\ldots x_k\ldots y_l$ if $x_k = y_0$. A \emph{cycle} is a path $c = x_0\ldots x_kx_0$ with $k\geq 1$, i.e., a cycle is a path from $x_0$ to itself with at least one other vertex along the path (a loop is not considered to be a cycle). A directed graph is called \emph{cyclic} if it contains a cycle, otherwise it is called \emph{acyclic}.

A directed graph is called \emph{strongly connected} if there is a path between each pair of vertices. A \emph{strongly connected component} of a directed graph is a maximal strongly connected subgraph.

\section{\hlr{Skipping Synthesis for Models with \CNMS Properties}}\label{sect:properties}
This section \hlrr{describes first} \hlrr{some} characteristics of several applications where synthesis does not add any restrictions besides those implied by the requirements. Then, \hlrr{it provides} properties that guarantee controllable, nonblocking, and maximally permissive supervisors that are together nonconflicting.

\subsection{Characteristics of models}
First, as the supervisors synthesized for the applications presented in~\cite{reijnen_supervisory_2017,reijnen_application_2018,reijnen_supervisory_2019} are intended to be implemented on control hardware, the input-output perspective of~\cite{balemi_control_1992} is used. This entails that each sensor is modeled by uncontrollable events, while actuators are modeled by controllable events. Each event represents a change of the state of such a component. This modeling paradigm results in a collection of numerous small plant models that do not share any events. Therefore, the plant model is a product system.

In the rest of this paper, an automaton \hlrr{is called} a \emph{sensor automaton} if its alphabet has only uncontrollable events, i.e., $\Sigma=\Sigma_u$, and an \emph{actuator automaton} if its alphabet has only controllable events, i.e., $\Sigma=\Sigma_c$.

Second, both sensors and actuators have cyclic behavior, often resulting in a trim and strongly connected plant model. For example, all sensors and actuators are modeled in this way in the production line in~\cite{reijnen_application_2018}. Furthermore, unreachable states in an uncontrolled plant represent states that are impossible to reach and are often not modeled or removed from the model.

Finally, requirements for applications often originate from safety risk analysis~\citep{modarres_risk_2016} and failure mode and effect analysis~\citep{stamatis_failure_2003}. States are identified in which some actuator actions would result in unsafe behavior. For example, the safety specifications of a waterway lock that need to be fulfilled by the supervisor are mentioned in Section 4.191 of~\cite{LBSv3}. Each of the 16 requirements given over there describes a state of the system and the disablement of certain actuator actions for that state. It is shown in~\cite{reijnen_supervisory_2017} that these textual specifications can be described with state-event invariant expressions.

\subsection{Properties}

The following properties together guarantee that the control problem itself is a modular globally controllable and nonblocking system.

\begin{definition}[\textbf{CNMS}]
A control problem $(\mathcal{P},\mathcal{R})$ satisfies \textbf{CNMS} (Controllable and Nonblocking Modular Supervisors properties) if it has the following properties:
\begin{enumerate}[label=\arabic*.]
    \item $\mathcal{P}$ is a product system.
    \item For all $P\in \mathcal{P}$ holds that $P$ is a strongly connected automaton with at least one marked state.
    \item For all $R\in \mathcal{R}$ holds
    \begin{enumerate}[label=\alph*.]
        \item $R$ is a state-event invariant expression $\mu \textbf{ needs } C$.
        \item $\mu\in \Sigma_c$.
        \item There exists no other requirement for this event $\mu$.
        \item $C$ is in a disjunctive normal form (see~\cite{davey_introduction_1990}) where each atomic proposition (or variable) is of the form $P.q$ with $P\in \mathcal{P}$.
        \item Each conjunction contains at most one reference to each $P\in \mathcal{P}$.
        \item When $P\in \mathcal{P}$ only has a single state, the literal $\neg P.q$ is not allowed in $C$.
        \item Each $P\in \mathcal{P}$ mentioned in $C$ is a sensor \hl{automaton}.
    \end{enumerate}
\end{enumerate}
\end{definition}

The intuition behind \hlr{the fact that} a system satisfying \textbf{CNMS} is controllable and nonblocking is as follows. Properties 1 and 2 ensure that the plant is already nonblocking in the open loop setting, i.e. without controller, and exhibits cyclic behavior. Furthermore, they ensure that individual plant models behave independently of the other plant models, i.e. an individual plant model can take a transition while the state of each of the other plant models remains the same.

Requirements satisfying Property 3 will not introduce blocking or controllability issues. There is no controllability issue, as there may not exist a requirement restricting the enablement of uncontrollable events. The reason why the controlled system is still nonblocking can be explained as follows. First, a sensor automaton can always go to a marked state with Properties 1, 2 and 3.b. For a plant automaton with one or more controllable events, \hlrr{it is known} from Properties 1 and 2 that from each state there exists a path to a marked state. For any controllable event along the path that is being restricted by a requirement, the condition of that requirement needs to be satisfied for the enablement of the controllable event. As only states of sensor automata are used in a condition and sensor automata can always reach each state without affecting other plant models, there exists a path in the sensor automata to satisfy the condition and subsequently enable the controllable event. By repeating the process of locally changing states in sensor automata, non-sensor automata can reach marked states if the requirements act as the supervisor.

The following theorem states that for a control problem satisfying \textbf{CNMS} synthesis can be skipped, i.e., the plant models and requirement models together already form controllable and nonblocking modular supervisors. In that case, the modular supervisor represented by the plant models and requirement models is by definition also maximally permissive. \hll{The proof of this theorem can be found in \ref{app:proofnosynthesis}.}

\begin{theorem}[\textbf{CNMS}~\citep{goorden_no_2019}]\label{thm:nosynthesis}
Let $(\mathcal{P},\mathcal{R})$ be a control problem satisfying \textbf{CNMS}. Then no supervisor synthesis is necessary, i.e., $\mathcal{P}\parallel \mathcal{R}$ is controllable and nonblocking, \hlll{hence also maximally permissive}.
\end{theorem}

\section{Dependency Graphs of Control Problems}\label{sect:graphanalysis}
As indicated in~\cite{goorden_no_2019}, there exist published control problems that do not satisfy \textbf{CNMS}, but \hllll{still} do not require synthesis. In this section, the \textbf{CNMS} properties are relaxed.

\subsection{Observations from models}\label{subsect:observations}
The main reason the control problems of~\cite{reijnen_supervisory_2017,reijnen_application_2018,reijnen_supervisory_2019} do not satisfy the \textbf{CNMS} properties is the violation of Property 3.g. In these control problems, there exist requirements that restrict the behavior of controllable events based on the behavior of plant models other than sensor \hl{automata}, which in turn may also be restricted by other requirements. Several causes of this violation are described below.

As pointed out in~\cite{zaytoon_synthesis_2001}, it may be desired to model the physical interaction between actuator and sensor components, because a supervisor that is proven to be deadlock-free for a model without interactions may deadlock after implementation on the physical system with interactions. Adding shared events to model the interactions will violate Property 1, as it is no longer a product system. Transforming this new model into a product system representation, the actuator and sensor models are combined into one due to the shared events. Therefore, requirements no longer refer only to states of sensor \hl{automata} (violating Property~3.g).

Second, sometimes a requirement needs to refer explicitly to the state of an actuator to guarantee correct behavior of the system. For example, consider a hydraulic arm that has one actuator to extend it and one actuator to retract it. In this case, the modeler could express that it is undesired that both actuators are on at the same time, resulting in two requirements each expressing that one actuator may only be activated if the other actuator is deactivated.

Finally, timer-based requirements violate Property 3.g. A timer is typically modeled with a controllable event to activate it and an uncontrollable event to indicate the timeout of the timer. Therefore, the model of a timer is neither a sensor \hl{automaton} nor an actuator \hl{automaton}. If a timer is needed, typically two requirements associated with it express when it can be activated (the controllable events of the timer model are used) and what should happen when the timer has timed out (the state of the timer model is used). Service calls in a server-client architecture are modeled in the same way, see for example~\cite{loose_component-wise_2018}, where service calls are modeled with controllable events and responses with uncontrollable events.

\subsection{Dependency graph}
\begin{figure}
	\begin{center}
		\begin{tikzpicture}[every node/.style={font=\sffamily\small, color=black}, 
			graph node/.style={circle,draw=lbl,font=\small,fill=lbl, inner sep=0pt, minimum size=3pt}, node distance=1cm, 
			->-/.style={thin, color=lbl, decoration={
				markings,
				mark=at position 0.6 with {\arrow{stealth}}}, postaction={decorate}}
			]
			\node[graph node, label={above:$P_1$}] (v1) {};
			\node (v) [below = of v1] {};
			\node[graph node, label={left:$P_2$}] (v2) [left  = of v] {};
			\node[graph node, label={right:$P_3$}] (v3) [right = of v] {};
			
			\draw[->-] (v1) -- node [left] {$e_1$} (v2);
			\draw[->-] (v1) -- node [left] {$e_2$} (v3);
		\end{tikzpicture}
	\end{center}
	\caption{The dependency graph $G_{\mathit{cp}}$ of control problem $(\{P_1,P_2,P_3\},\{R\})$ with $R = \mu \textbf{ needs } P_2.q_1 \vee \neg P_3.q_1$ and $\mu\in \Sigma_{P_1}$. This graph has three vertices $P_1, P_2$, and $P_2$ and two edges $e_1$ and $e_2$.}
	\label{fig:exampleGraph}
\end{figure}
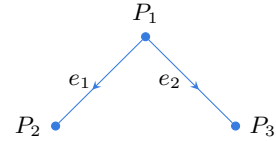

Control problems $(\mathcal{P},\mathcal{R})$ satisfying all properties of \textbf{CNMS} \emph{except} Property 3.g \hllll{are called to satisfy} the Relaxed Controllable and Nonblocking Modular Supervisors properties (\textbf{RCNMS}). \hllll{For control problems satisfying \textbf{CNMS} or \textbf{RCNMS},} a directed graph can be constructed indicating the dependencies between plant models from $\mathcal{P}$ based on the requirement models from $\mathcal{R}$. In this directed graph, each vertex represents a plant model from the control problem. For each requirement in the control problem, a set of edges is present in the graph such that the initial vertex of each edge is the plant model containing the event that is restricted by the requirement. Furthermore, for each plant model used in the condition of the requirement there is an edge having this plant model as terminating vertex. For example, consider the control problem $(\{P_1,P_2,P_3\},\{R\})$ with $R = \mu \textbf{ needs } P_2.q_1 \vee \neg P_3.q_1$ and $\mu\in \Sigma_{P_1}$. The dependency graph of this control problem is shown in Figure~\ref{fig:exampleGraph}. It has three vertices $P_1, P_2$ and $P_3$. For requirement $R$, two edges $e_1$ and $e_2$ are present such that $\mbox{init}(e_1)=\mbox{init}(e_2) = P_1$, as the restricted event of $R$ originates from $P_1$, $\mbox{ter}(e_1) = P_2$, as $P_2$ is mentioned in the condition of $R$, and $\mbox{ter}(e_2) = P_3$, as $P_3$ is mentioned in the condition of $R$. This example also shows that there may be multiple, but isomorphic, dependency graphs for the same control problem.

More formally, let the \emph{dependency graph} of control problem $(\mathcal{P},\mathcal{R})$ be a directed graph $G_{\mathit{cp}} = (\mathcal{P},E)$ \hlr{s.t.} $\hlr{\forall} R\in \mathcal{R}$ a set of edges $E_R\subseteq E$ is constructed \hlr{s.t.} $\hlr{\forall} e\in E_R$, $\mbox{init}(e) = P_i\in \mathcal{P} \hlr{\wedge} \mathit{event}(R)\in\Sigma_{P_i}$, and $\hlr{\forall} P_j\in \mathcal{P}$ used in $\mathit{cond}(R)$ $\hlr{\exists}e\in E_R$ with $\mbox{ter}(e) = P_j$, and finally $E = \bigcup_{R\in \mathcal{R}} E_R$.

A control problem satisfying \textbf{CNMS} results in an acyclic bipartite dependency graph.

\section{\hlr{Skipping Synthesis with Dependency Graphs}}\label{sect:acyclicgraphs}

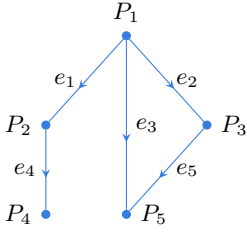
\begin{figure}
\begin{center}
\begin{tikzpicture}[every node/.style={font=\sffamily\small, color=black}, graph node/.style={circle,draw=lbl,font=\small,fill=lbl, inner sep=0pt, minimum size=3pt}, node distance=1cm, 
	->-/.style={thin, color=lbl, decoration={
			markings,
			mark=at position 0.6 with {\arrow{stealth}}}, postaction={decorate}}
	]
\node[graph node, label={above:$P_1$}] (v1) {};
\node (v) [below = of v1] {};
\node[graph node, label={left:$P_2$}] (v2) [left  = of v.center] {};
\node[graph node, label={right:$P_3$}] (v3) [right = of v.center] {};
\node[graph node, label={right:$P_5$}] (v4) [below = of v] {};
\node[graph node, label={left:$P_4$}] (v5) [left  = of v4.center] {};

\draw[->-] (v1) -- node [left]  {$e_1$} (v2);
\draw[->-] (v1) -- node [right] {$e_2$} (v3);
\draw[->-] (v1) -- node [right] {$e_3$} (v4);
\draw[->-] (v3) -- node [right] {$e_5$} (v4);
\draw[->-] (v2) -- node [left] {$e_4$} (v5);
\end{tikzpicture}
\end{center}
\caption{A dependency graph $G_{\mathit{cp}}$ of a control problem with five plant models satisfying \textbf{RCNMS} where $P_4$ and $P_5$ are sensor \hl{automata}.}
\label{fig:exampleGraph2}
\end{figure}

Figure~\ref{fig:exampleGraph2} shows the dependency graph of a control problem satisfying \textbf{RCNMS}, but not \textbf{CNMS}.
Plant models $P_2$ and $P_3$ have both incoming and outgoing edges, which indicate that the enablement of one or more events in each plant model is restricted by a requirement and that one or more states of each plant model are used in the condition of a requirement. Therefore, this model does not satisfy Properties 3.g of \textbf{CNMS}.
This example demonstrates why the control problem underlying this graph still requires no synthesis.

For the \textbf{CNMS} property, \hlrr{it is} shown with Theorem~\ref{thm:nosynthesis} that, essentially, no edge is permanently disabled. As the properties ensure that in a controlled system each sensor \hl{automaton} can always reach each state, the condition of each state-event invariant expression can be eventually satisfied, enabling the controllable event of each state-event invariant expression. Therefore, each non-sensor plant model can reach all states from each state.

This argument can be used inductively to show that a control problem satisfying \textbf{RCNMS} still requires no synthesis. As the behavior of plants $P_2$ and $P_3$ in Figure~\ref{fig:exampleGraph2} only depends on sensor plants $P_4$ and $P_5$, it holds that $P_2$ and $P_3$ can reach all states from each state. Since the behavior of $P_1$ only depends on the plant models $P_2$, $P_3$, and $P_5$, and it is already known that all these models can reach all states from each state, \hlrr{it} can conclude that $P_1$ also can reach all states from each state. Therefore, the complete control problem is \hlll{controllable, nonblocking, and maximally permissive}. This is formalized in Theorem~\ref{thm:nosynthesisrelaxed}. \hll{The proof of this theorem can be found in \ref{app:proofnosynthesisrelaxed}.} \hlll{The result of this theorem is that if a control problem satisfies \textbf{RCNMS}, then no synthesis is needed.}

\begin{theorem}[Acyclic \textbf{RCNMS}]\label{thm:nosynthesisrelaxed}
Let $(\mathcal{P},\mathcal{R})$ be a control problem satisfying \textbf{RCNMS}. \hlll{If the dependency graph $G_{\mathit{cp}}$ of $(\mathcal{P},\mathcal{R})$ is acyclic and \hlll{self}-loop free, then $\mathcal{P}\parallel \mathcal{R}$ is controllable and nonblocking, hence also maximally permissive.}
\end{theorem}

\section{\hlr{Sectionalizing Control Problems}}\label{sect:cyclicgraphs}
\begin{figure}
\begin{center}
\begin{tikzpicture}[->,>=stealth',shorten >=1pt,main node/.style={circle,draw,font=\small}, node distance=1.5cm]
\node[main node, label={below:$q_1$}, initial, initial text={$P_1$},initial where=left, accepting] (l1) {};
\node[main node, label={below:$q_2$}]  (l2) [right = of l1] {};
\node[main node, label={below:$q_3$}, initial, initial text={$P_2$},initial where=left, accepting] (l3) [below = of l1] {};
\node[main node, label={below:$q_4$}] (l4) [right = of l3] {};


\path[every node/.style={font=\sffamily\small}, bend left=15]
    (l1) edge node [above] {$a$} (l2)
    (l2) edge node [below] {$b$} (l1)
    (l3) edge node [above] {$c$} (l4)
    (l4) edge node [below] {$d$} (l3)
;

\node (r1) [right = 1cm of l2] {$R_1: b \textbf{ needs } P_2.q_4$};
\node (r2) [below = .1cm of r1] {$R_2: d \textbf{ needs } P_1.q_2$};
\node (r3) [right = 1cm of l4] {$R_3: a \textbf{ needs } P_2.q_3$};
\node (r4) [below = .1cm of r3] {$R_4: c \textbf{ needs } P_1.q_1$};


\end{tikzpicture}
\end{center}
\caption{Both control problems $\mathit{CP}_1=(\{P_1, P_2\},\{R_1, R_2\})$ and $\mathit{CP}_2=(\{P_1, P_2\},\{R_3, R_4\})$ result in cyclic dependency graphs, but the first contains a blocking issue while the second one \hlr{does} not.}
\label{fig:exampleCycle}
\end{figure}
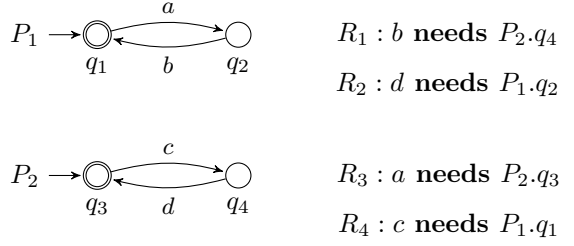

For the case that a dependency graph is cyclic, supervisory control synthesis may be needed as $\mathcal{P}\parallel \mathcal{R}$ could be blocking. Figure~\ref{fig:exampleCycle} shows two control problems $\mathit{CP}_1=(\{P_1, P_2\},\{R_1, R_2\})$ and $\mathit{CP}_2=(\{P_1, P_2\},\{R_3, R_4\})$, both based on the same set of plant models $\{P_1, P_2\}$. Those control problems result in the same cyclic dependency graph. However, $\mathit{CP}_1$ is blocking, while $\mathit{CP}_2$ is nonblocking.

\hlll{Similarly, a dependency graph with a self-loop might also indicate that $\mathcal{P}\parallel \mathcal{R}$ is blocking. Consider again control problem $\mathit{CP}_1$, but now requirement $R_2$ is replaced by $d \textbf{ needs } P_2.q_3$. This results in a self-loop in the dependency graph and the control problem is blocking. Yet, if requirement $R_2$ is replaced by $d \textbf{ needs } P_2.q_4$, the dependency graph still has a self-loop, but the control problem is nonblocking.}

So, a dependency graph containing cycles or self-loops may or may not require synthesis to obtain a controllable, nonblocking, and maximally permissive supervisor. In the remainder of this section \hlrr{it is shown} that in case of a cyclic dependency graph the original control problem can be reduced to partial control problems containing the cycles.

\subsection{Control problem reduction}

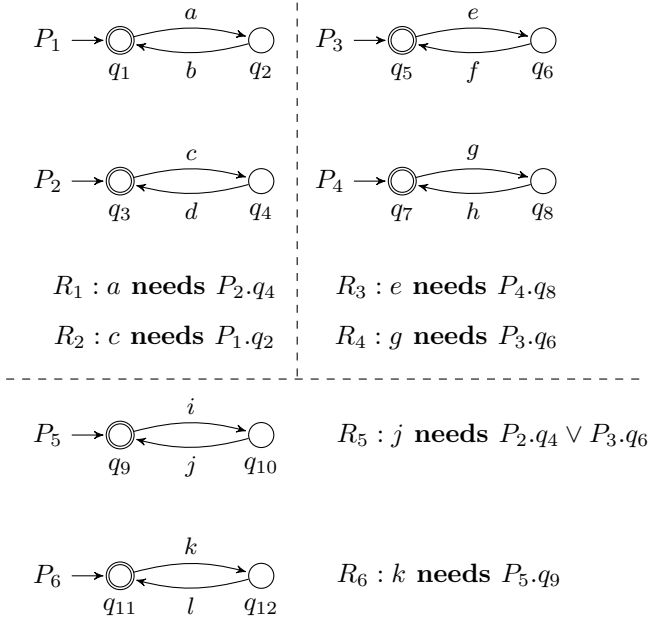
\begin{figure}
\begin{center}
\begin{tikzpicture}[->,>=stealth',shorten >=1pt,main node/.style={circle,draw,font=\small}, node distance=1.5cm]
\node[main node, label={below:$q_1$}, initial, initial text={$P_1$},initial where=left, accepting] (l1) {};
\node[main node, label={below:$q_2$}]  (l2) [right = of l1] {};
\node[main node, label={below:$q_3$}, initial, initial text={$P_2$},initial where=left, accepting] (l3) [below = of l1] {};
\node[main node, label={below:$q_4$}] (l4) [right = of l3] {};

\node[main node, label={below:$q_5$}, initial, initial text={$P_3$},initial where=left, accepting] (l5) [right = of l2] {};
\node[main node, label={below:$q_6$}]  (l6) [right = of l5] {};
\node[main node, label={below:$q_7$}, initial, initial text={$P_4$},initial where=left, accepting] (l7) [below = of l5] {};
\node[main node, label={below:$q_8$}] (l8) [right = of l7] {};

\path[every node/.style={font=\sffamily\small}, bend left=15]
    (l1) edge node [above] {$a$} (l2)
    (l2) edge node [below] {$b$} (l1)
    (l3) edge node [above] {$c$} (l4)
    (l4) edge node (d) [below] {$d$} (l3)
    (l5) edge node [above] {$e$} (l6)
    (l6) edge node [below] {$f$} (l5)
    (l7) edge node [above] {$g$} (l8)
    (l8) edge node (h) [below] {$h$} (l7)
;

\node (r1) at ([xshift=-10]d.south) [below=5mm] {$R_1: a \textbf{ needs } P_2.q_4$};
\node (r2) [below = 1mm of r1] {$R_2: c \textbf{ needs } P_1.q_2$};
\node (r3) at ([xshift=-10]h.south) [below=5mm] {$R_3: e \textbf{ needs } P_4.q_8$};
\node (r4) [below = 1mm of r3] {$R_4: g \textbf{ needs } P_3.q_6$};

\node (a) [below left  = 2.5cm and 1.5cm of l3.center] {};
\node (b) [right = 8.5cm of a.center] {};
\draw[-,dashed] (a) -- (b);

\path (l2) -- node [pos=.2] (n) {} (l5);
\draw[-,dashed] (n) +(0,5mm) -- (intersection of a--b and n--[yshift=1pt]n);

\node[main node, label={below:$q_9$}, initial, initial text={$P_5$},initial where=left, accepting] (l9) [below = 3cm of l3] {};
\node[main node, label={below:$q_{10}$}] (l10) [right = of l9] {};
\node[main node, label={below:$q_{11}$}, initial, initial text={$P_6$},initial where=left, accepting] (l11) [below = of l9] {};
\node[main node, label={below:$q_{12}$}] (l12) [right = of l11] {};

\path[every node/.style={font=\sffamily\small}, bend left=15]
    (l9)  edge node [above] {$i$} (l10)
    (l10) edge node [below] {$j$} (l9)
    (l11) edge node [above] {$k$} (l12)
    (l12) edge node [below] {$l$} (l11)
;

\node (r1) [right = .7cm of l10] {$R_5: j \textbf{ needs } P_2.q_4 \vee P_3.q_6$};
\node (r1) [right = .7cm of l12] {$R_6: k \textbf{ needs } P_5.q_9$};
\end{tikzpicture}
\end{center}
\caption{A control problem $\mathit{CP}=(\{P_1,\ldots,P_6\},\{R_1,\ldots,R_6\})$.}
\label{fig:exampleVertexSetAnalysis}
\end{figure}

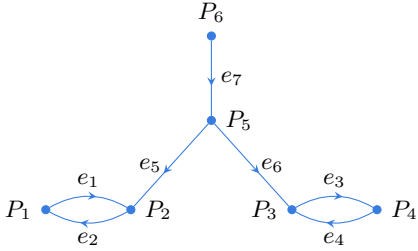
\begin{figure}
\begin{center}
\begin{tikzpicture}[every node/.style={font=\sffamily\small, color=black}, graph node/.style={circle,draw=lbl,font=\small,fill=lbl, inner sep=0pt, minimum size=3pt}, node distance=1cm, 
	->-/.style={thin, color=lbl, decoration={
			markings,
			mark=at position 0.6 with {\arrow{stealth}}}, postaction={decorate}}
	]
\node[graph node, label={left:$P_1$}] (v1) {};
\node[graph node, label={right:$P_2$}] (v2) [right = of v1] {};
\node[graph node, label={left:$P_3$}] (v3) [right = 2cm of v2] {};
\node[graph node, label={right:$P_4$}] (v4) [right = of v3] {};

\path (v2) -- node [pos=.5] (v) {} (v3);
\node[graph node, label={right:$P_5$}] (v5) [above = of v]  {};
\node[graph node, label={above:$P_6$}] (v6) [above = of v5] {};

\draw[->-] (v1) to[bend left] node [above] {$e_1$} (v2);
\draw[->-] (v2) to[bend left] node [below] {$e_2$} (v1);
\draw[->-] (v3) to[bend left] node [above] {$e_3$} (v4);
\draw[->-] (v4) to[bend left] node [below] {$e_4$} (v3);

\draw[->-] (v5) -- node [left] {$e_5$} (v2);
\draw[->-] (v5) -- node [right] {$e_6$} (v3);
\draw[->-] (v6) -- node [right] {$e_7$} (v5);


\end{tikzpicture}
\end{center}
\caption{The dependency graph of the control problem shown in Figure~\ref{fig:exampleVertexSetAnalysis}.}
\label{fig:exampleGraphVertexSetAnalysis}
\end{figure}

From the dependency graph, all strongly connected components containing a cycle are identified. For each strongly connected component, the set of vertices (plant models) is denoted by $\phi$, and the collection of these sets is denoted by $\Phi = \{\phi_1,\ldots,\phi_m\}$. From the definition of strongly connected components, it follows that they are non-overlapping. Figure~\ref{fig:exampleVertexSetAnalysis} shows control problem $\mathit{CP}$, with its dependency graph $G_{\mathit{CP}}$ shown in Figure~\ref{fig:exampleGraphVertexSetAnalysis}. $G_{\mathit{CP}}$ contains two cycles $c_1=P_1P_2P_1$ and $c_2=P_3P_4P_3$, and the strongly connected components of these two cycles are $\phi_1 = \{P_1, P_2\}$ and $\phi_2 = \{P_3, P_4\}$.

This example also shows plants whose behavior depends on the behavior of these strongly connected components. Requirement $R_5$ restricts the behavior of component model $P_5$ based on the behavior of component models $P_2$ and $P_3$. In this example, a supervisor is needed, as any synthesized supervisor for requirements $R_1, R_2, R_3$, and $R_4$ would make states $P_2.q_4$ and $P_3.q_6$ unreachable in the closed-loop system, and therefore requirement $R_5$ never enables event $j$. A supervisor is needed to disable event $i$ to prevent component $P_5$ from being blocked in state $q_{10}$. Therefore, it is insufficient to only analyze the strongly connected components in isolation.

\hlr{As a next step}, vertices are added recursively to these strongly connected components. A vertex is added to a set of vertices if there exists an edge such that this edge originates in this added vertex and terminates in one of the vertices already in the set. Eventually, the strongly connected component is enlarged with those vertices from which there exists a path to a vertex in the strongly connected component. Formally, the extended set of vertices for each strongly connected component with a cycle $\phi_i$, denoted by $V_{\phi_i}$, is defined as $V_{\phi_i} = \{P\in \mathcal{P} \ |\ \exists p = x_0x_1\ldots x_k, k \geq 0,  p \in \mathit{Path}(G_{\mathit{CP}}) \mbox{ s.t. } x_0 = P \wedge x_k\in\phi_i\}$, and $\mathbb{V} = \{V_{\phi_1}, \ldots, V_{\phi_m}\}$, with $\mathit{Path}(G_{\mathit{CP}})$ the set of all paths in $G_{\mathit{CP}}$. The extended sets of vertices for the example are calculated as $V_{\phi_1} = \{P_1, P_2, P_5, P_6\}$ and $V_{\phi_{2}} = \{P_3, P_4, P_5, P_6\}$.

Still, it is insufficient to only analyze each extended vertex set $V_{\phi_i}$. Two extended vertex sets may share vertices. This sharing could be problematic. In the running example, $V_{\phi_1}$ and $V_{\phi_2}$ share vertices $P_5$ and $P_6$.

Shared vertices between extended sets $V_{\phi_i}$ and $V_{\phi_j}$ will not always imply that it is necessary to analyze the partial control problem represented by $V_{\phi_i} \cup V_{\phi_j}$. Sometimes, it is still sufficient to analyze the partial control problems of $V_{\phi_i}$ and $V_{\phi_j}$ separately. For the control problem $\mathit{CP}$ of Figure~\ref{fig:exampleVertexSetAnalysis}, $V_{\phi_1}$ and $V_{\phi_2}$ should be combined, as the edges $e_5$ and $e_6$ relate to the same requirement $R_5$. The evaluation of the condition of requirement $R_5$ requires the result of the analysis of both strongly connected components $\phi_1$ and $\phi_2$. If requirement $R_5$ \hlrr{is replaced} by, for example, the two requirements $R_5': j \textbf{ needs } P_2.q_4$ and $R_5'': j \textbf{ needs } P_3.q6$, the extended sets $V_{\phi_i}$ and $V_{\phi_j}$ do \emph{not} need to be merged for analyzing the cycles. While the dependency graph remains the same, edges $e_5$ and $e_6$ are now induced by different requirements.

Unfortunately, the above reasoning cannot be generalized. The control problem in Figure~\ref{fig:exampleVertexSetAnalysis} \hlrr{is modified} again. An additional transition is added to plant model $P_5$ from state $q_{10}$ to $q_9$ labeled with $j'$. Requirement $R_5$ is replaced by two requirements $R_5': j \textbf{ needs } P_2.q_4$ and $R_5'': j'\textbf{ needs } P_3.q_6$. Again, the dependency graph in Figure~\ref{fig:exampleGraphVertexSetAnalysis} remains unchanged. The controllable, nonblocking, and maximally permissive supervisor $S_1$ synthesized for the partial control problem $(\{P_1,P_2,P_5,P_6\},\{R_1,R_2,R_5',R_6\})$ would disable the transition labeled with event $j$, and the controllable, nonblocking, and maximally permissive supervisor $S_2$ synthesized for the partial control problem $(\{P_3,P_4,P_5,P_6\},\{R_3,R_4,R_5'',R_6\})$ would disable the transition labeled with event $j'$. Now, $S_1\parallel S_2$ is blocking, because plant $P_5$ deadlocks in state $q_{10}$, as the supervisors together disable both event $j$ and event $j'$.

Therefore, two extended sets of vertices need to be merged once they share a vertex. Let $\sim\ \subseteq \mathbb{V}\times\mathbb{V}$ be a relation between extended sets of vertices. $(V_{\phi_i},V_{\phi_j})\in\ \sim$ if and only if $V_{\phi_i}\cap V_{\phi_j}\neq \emptyset$, i.e., they share at least one vertex, \hlr{or $(V_{\phi_i},V_{\phi_k})\in\ \sim$ and $(V_{\phi_k}, V_{\phi_j})\in\ \sim$ for some $V_{\phi_k}$, i.e., they share a vertex with a common extended set of vertices}. From this definition, it follows directly that $\sim$ is \hlr{an equivalence relation, as it is reflexive, symmetric, and transitive}.

Now, the partition $\mathbb{W}$ of $\mathbb{V}$ is the set of all equivalence classes of $\mathbb{V}$ with equivalence relation $\sim$, i.e., $\mathbb{W} = \mathbb{V}/ \sim$ is the quotient set of $\mathbb{V}$ by $\sim$. For the example shown in Figure~\ref{fig:exampleGraphVertexSetAnalysis}, the partition $\mathbb{W}$ is $\{\{P_1,\ldots,P_6\}\}$.

A \emph{simplified partial control problem} $(P'_s,\tilde{R}_s)$ \emph{represented by a subset of vertices} $P'_s\subseteq \mathcal{P}$ is constructed as follows. First, $R'_s = \{R\in \mathcal{R}\ |\ \exists P\in P'_s \mbox{ s.t. } \mathit{event}(R)\in\Sigma_{P}\}$. Subsequently, the condition of each requirement in this set is adjusted where each literal containing reference to a state of a plant \emph{not} in $P'_s$ is replaced by the boolean literal $\textbf{T}$, resulting in the set of adjusted requirements $\tilde{R}_s$.

Theorem~\ref{thm:cyclicgraphs} contains the main result of this section: based on the dependency graph, synthesizing a supervisor can be performed following a modular approach which guarantees (global) controllability, nonblockingness, and maximal permissiveness. This theorem can be used to reduce the computational complexity of supervisor synthesis. \hll{The proof of this theorem can be found in \ref{app:proofcyclicgraphs}.}

\begin{theorem}[Cyclic \textbf{RCNMS}]\label{thm:cyclicgraphs}
Let $(\mathcal{P},\mathcal{R})$ be a control problem satisfying \textbf{RCNMS} and let $G_{\mathit{cp}}$ be its dependency graph. For each $W\in\mathbb{W}$, let $S_{W}$ be a controllable, nonblocking, and maximally permissive supervisor for the simplified partial control problem represented by $\bigcup_{V\in W} V$. Then $\mathcal{P} \parallel \mathcal{R} \parallel (\parallel_{W\in\mathbb{W}} S_{W})$ is a modular, controllable, nonblocking, and maximally permissive supervisor of $(\mathcal{P},\mathcal{R})$.
\end{theorem}

Theorem~\ref{thm:cyclicgraphs} shows for which partial control problems synthesis might still be needed and for which part of the system no synthesis is needed. In the worst-case situation, the original control problem is the only single equivalence class in $\mathbb{W}$ \hlr{and no reduction can be achieved with the presented method}. Sections~\ref{sect:casestudy} and~\ref{sect:tunnel} show that there exist industrial systems for which the control problem can be reduced. There are two options available for those partial control problems that might need synthesis: either synthesize a supervisor with an existing synthesis algorithm, like monolithic~\citep{ramadge_control_1989}, compositional~\citep{mohajerani_framework_2014}, and incremental synthesis~\citep{brandin_incremental_2004}, or reason with an additional method that synthesis is still not needed (as it is known for the case studies in~\cite{reijnen_supervisory_2017, reijnen_application_2018, reijnen_supervisory_2019} that no synthesis is needed). The second option is left open for future work.

\section{\hlr{Case Study 1: FESTO production line}}\label{sect:casestudy}

\begin{figure}
\center
\includegraphics[width=.9\linewidth]{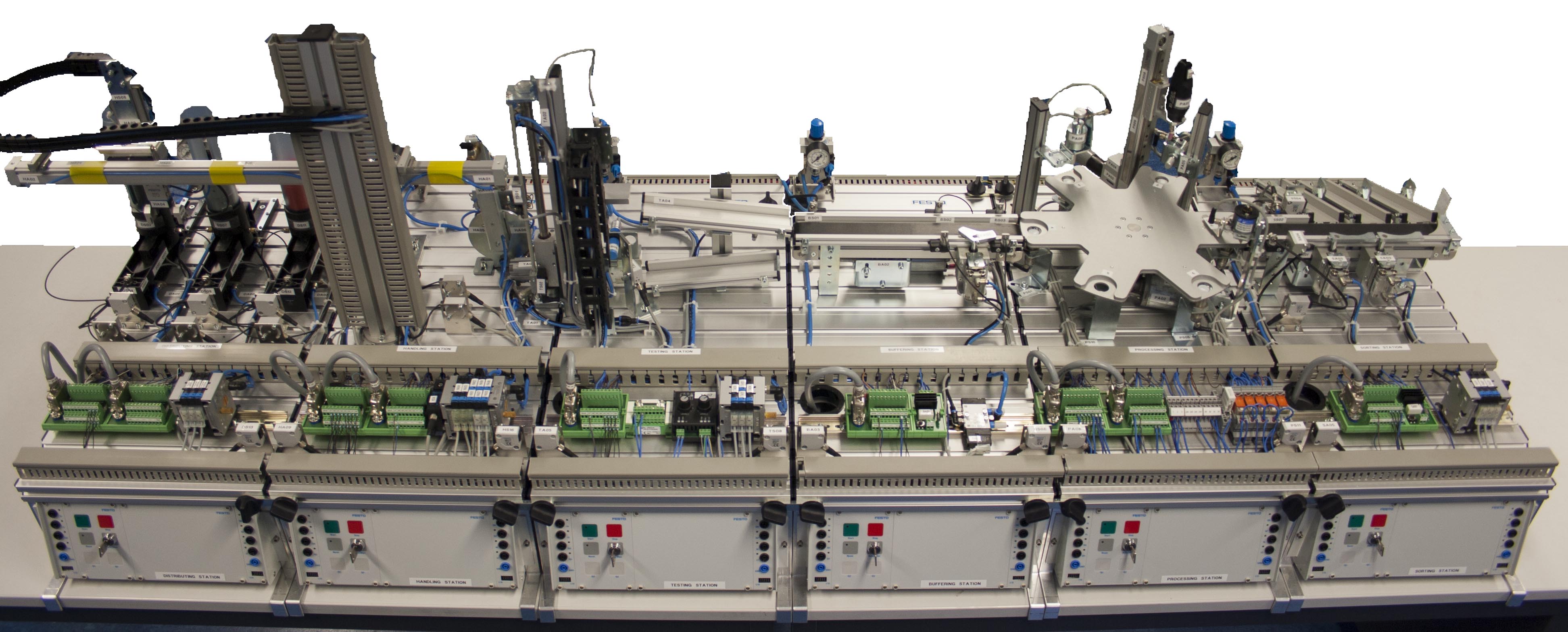}
\caption{Overview of the FESTO production line.\label{fig:festo}}
\end{figure}

In this section, the proposed method is demonstrated with a case study. For this case study, a small-scale production line consisting of six workstations has been considered, see Figure~\ref{fig:festo}. The hardware of the system is produced by Festo Didactic for vocational training in the field of industrial automation. This system has been previously modeled in~\cite{reijnen_application_2018}. In the remainder of this section, first a description of this production line \hlrr{is provided}. Subsequently, two workstations in isolation \hlrr{are analyzed} to demonstrate Theorems~\ref{thm:nosynthesis} and~\ref{thm:nosynthesisrelaxed}. Finally, the complete production line is analyzed to demonstrate Theorem~\ref{thm:cyclicgraphs}.

\subsection{Case description}\label{sect:casedescription}
While no real production is taking place, all movements, velocities, and timings are as if it were. In total, the production line consists of 28 actuators, like DC motors and pneumatic cylinders, and 59 sensors, like capacitive, optical, and inductive ones.

The intended controlled behavior is as follows. Products enter the production line through the \emph{distribution} station where they have been placed in three storage tubes. For each storage tube, a pusher is able to release a new product. The second workstation, the \emph{handling} station, transports products from the distribution station to the testing station in two steps. First, a pneumatic gripper transports released products to an intermediate buffer. From this buffer, a transfer cylinder picks them up and places them in the \emph{testing} station where the product height is tested. Correct products are moved by an air slide to the next station while rejected products are stored in a local buffer. In the fourth station, the \emph{buffering} station, products can be held on a conveyor belt. A separator controls the release of products from the conveyor belt. At the next station, the \emph{processing} station, products are processed. A turntable with six places rotates products through this station. After entering the processing station, the product is moved to a testing location where the orientation of the product is checked. Subsequently, at the next location a hole is drilled in the product only if the orientation is correct. At the fourth location, processed products are ejected to the sorting station. The last two locations can be used to correct the orientation if needed, and in that case the product can be processed again. In the final workstation, the \emph{sorting} station, products are stored on one of the three slides, depending on color and the material of the product. Two pneumatic gates can be used to divert the product to the correct slide.

In~\cite{reijnen_application_2018}, a model of the production line is presented, which is slightly modified for this case study to have exclusively state-event invariant expressions; adjustments are indicated by comments in the model. The model contains 75 plant models and 214 requirement models, which can be accessed at a GitHub repository\footnote{\url{https://github.com/magoorden/NonblockingModularSupervisors}}.

Performing monolithic synthesis on this model reveals that the synthesized supervisor does not impose any additional restrictions to ensure controllable and nonblocking behavior, i.e., the control problem can already act as a modular, controllable, nonblocking, and maximally permissive supervisor.

\subsection{Distribution station}
The distributed construction of the model of the production line eases the individual analysis of workstations. To start with, the distribution station is analyzed.

\begin{figure}
\center
\includegraphics[width=.75\linewidth]{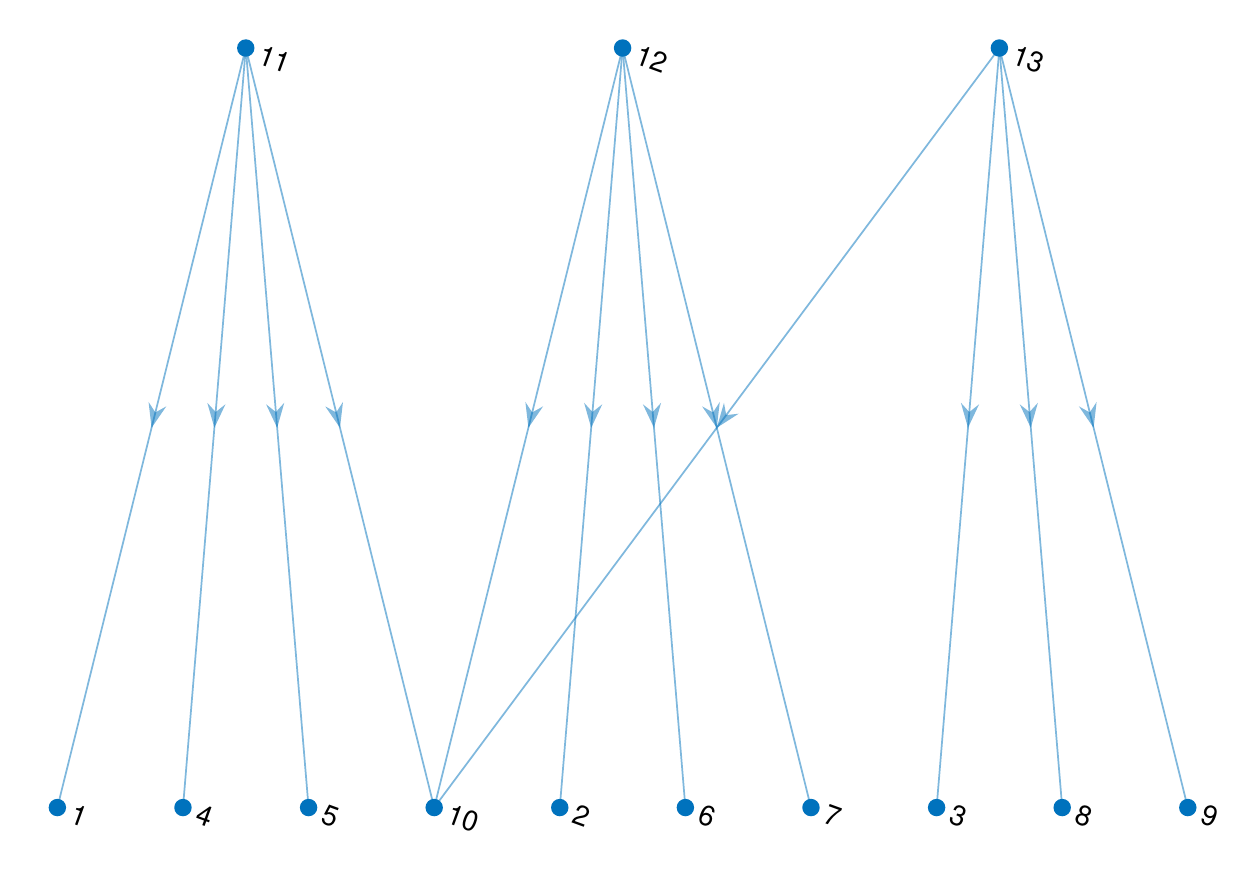}
\caption{The dependency graph of the distribution station. For readability, numbers are used as nodes instead of the actual names from the model.\label{fig:distribution}}
\end{figure}

Figure~\ref{fig:distribution} shows the dependency graph of the distribution station. To prevent cluttering of names, numbers are displayed in this and subsequent figures instead of the actual plant names in the model. The readme file in the model repository explains how the actual names can be obtained. Plant models 1 through 10 are sensor \hl{automata}, i.e., they only have uncontrollable events in their alphabet, plant models 11, 12, and 13 are actuator \hl{automata}, i.e., they only have controllable events in their alphabet. As each edge in this dependency graph has an actuator \hl{automaton} as initial vertex and a sensor \hl{automaton} as terminal vertex, Theorem~\ref{thm:nosynthesis} applies. This indicates that, if a supervisor is only needed for this workstation, synthesis can be skipped and the control problem already represents the supervisor.

\subsection{Sorting station}
\begin{figure}
\center
\includegraphics[width=.9\linewidth]{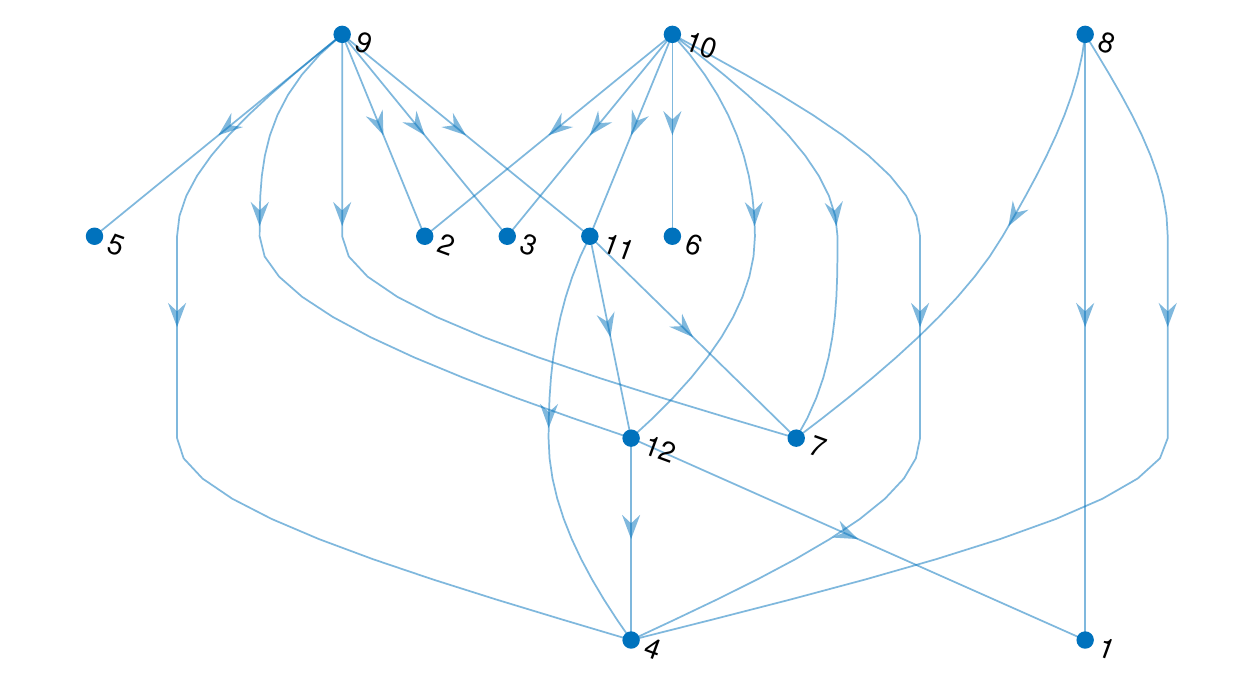}
\caption{The dependency graph of the sorting station.\label{fig:sorting}}
\end{figure}

Figure~\ref{fig:sorting} shows the dependency graph of the sorting station. In this workstation, plant models 1 through 7 are sensor \hl{automata}, plant models 8 through 11 are actuator \hl{automata} and plant model 12 contains both controllable and uncontrollable events. This graph already indicates that Theorem~\ref{thm:nosynthesis} does not apply: there are edges (representing requirements) that have a non-sensor \hl{automaton} as a terminal vertex. In particular, plant models 11 and 12 have both incoming and outgoing edges, which indicates a violation of Property 3.g of the \textbf{CNMS} properties. Fortunately, as the model satisfies the \textbf{RCNMS} properties and the control dependency graph is acyclic, Theorem~\ref{thm:nosynthesisrelaxed} applies. Therefore, synthesis can be skipped.

\subsection{Production line}
\begin{figure}
\center
\includegraphics[width=\linewidth]{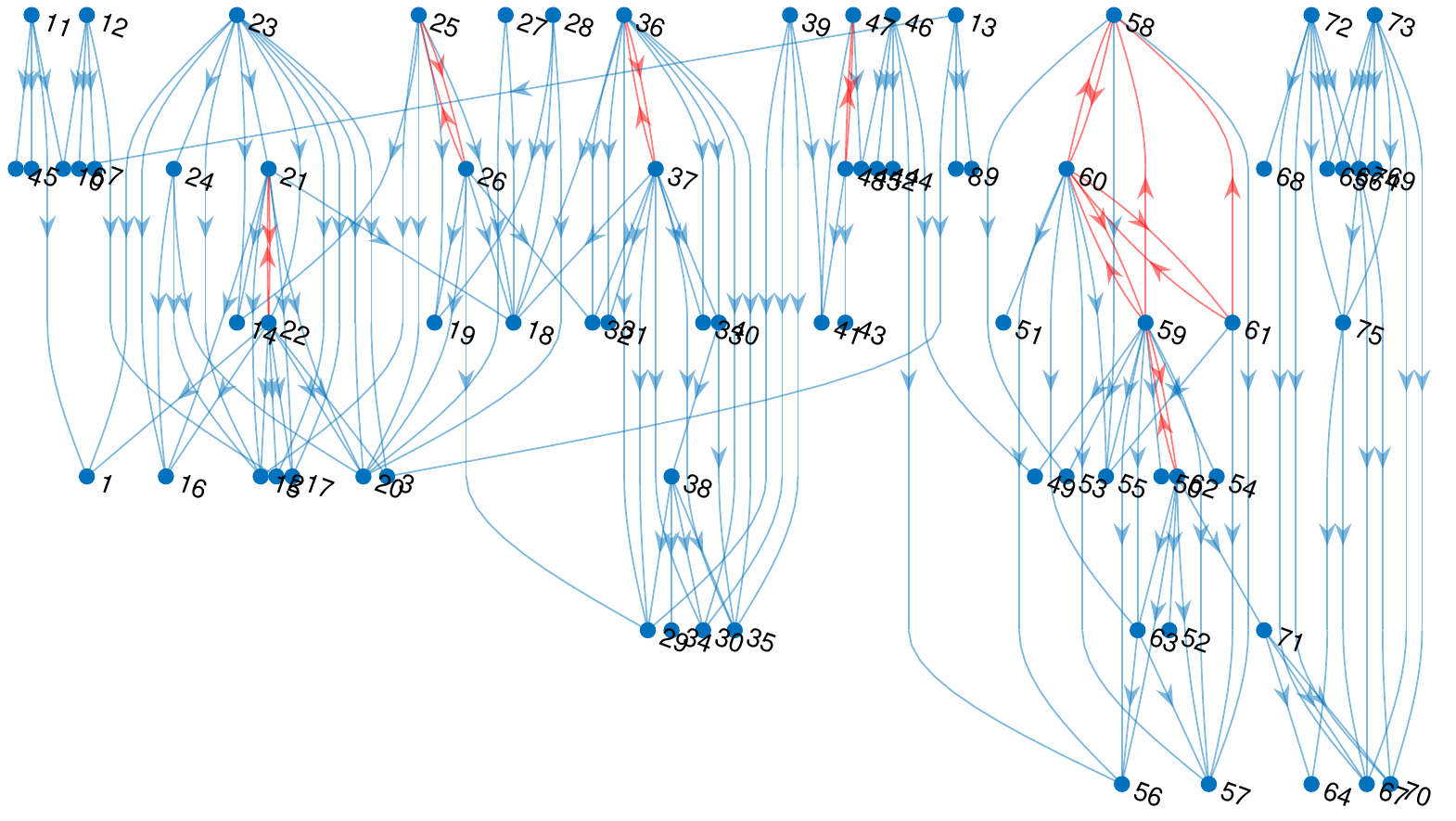}
\caption{The dependency graph of the complete production line. Red \hlr{color} indicates cycles.\label{fig:completefesto}}
\end{figure}

Figure~\ref{fig:completefesto} shows the dependency graph of the complete production line. Cycles in this graph are indicated in red. Clearly, both Theorems~\ref{thm:nosynthesis} and~\ref{thm:nosynthesisrelaxed} are not applicable for the control problem of the complete production line.

With the help of Theorem~\ref{thm:cyclicgraphs}, the problem of synthesizing a monolithic supervisor can be reduced to analyzing smaller control problems based on the identified cycles. In the dependency graph, five strongly connected components containing cycles \hlrr{can be identified}: $\phi_1 = \{P_{21}, P_{22}\}, \phi_2 = \{P_{25}, P_{26}\}$, $\phi_3 = \{P_{36}, P_{37}\}$, $\phi_4 = \{P_{47}, P_{48}\}$, and $\phi_5 = \{P_{58}, P_{59}, P_{60}, P_{61}, P_{62}\}$. Next, these sets need to be extended to include all plant models from which there exists a path to one of the plants in that particular strongly connected component. This is only the case for $\phi_1$, as from $P_{23}$ there exists a path from $P_{23}$ to $P_{21}$ (and $P_{22}$). Therefore, $V_{\phi_1} = \{P_{21}, P_{22}, P_{23}\}$, while $V_{\phi_2} = \phi_2$, $V_{\phi_3} = \phi_3$, $V_{\phi_4} = \phi_4$, and $V_{\phi_5} = \phi_5$. In this case, there is no overlap between these extended sets, so $W_i = V_{\phi_i}$ for $i\in [1,5]$.

Finally, five supervisors, $S_1,\ldots, S_5$ are synthesized, one for each simplified partial control problem represented by $\bigcup_{V_{\phi_i}\in W} V_{\phi_i}$. From Theorem~\ref{thm:cyclicgraphs} it follows that $\mathcal{P}\parallel \mathcal{R} \parallel S_1\parallel S_2\parallel S_3\parallel S_4\parallel S_5$ is a modular, controllable, nonblocking, and maximally permissive supervisor for the production line.

\begin{table}
\renewcommand{\arraystretch}{1.5}
\vspace{0.5em}
\caption{Results of supervisory control synthesis for the production line.}
\label{tab:results}
\centering
{\tabulinesep=1.2mm
\begin{tabu}{|X[1,l,m]|X[1,r,m]|X[1,r,m]|X[1,r,m]|}
  \tabucline-
  \rowfont[l]{}
  Model & \raggedright  Uncontrolled state-space size & Controlled state-space size & Synthesis duration [s]\\ \tabucline[1.5pt]-
  Monolithic & $5.9\cdot 10^{26}$ & $2.2\cdot 10^{25}$ & 370 \\ \tabucline-
  $S_1$ & 8 & 6 & $<1$ \\ \tabucline-
  $S_2$ & 4 & 3 & $<1$ \\ \tabucline-
  $S_3$ & 4 & 3 & $<1$ \\ \tabucline-
  $S_4$ & 6 & 6 & $<1$ \\ \tabucline-
  $S_5$ & 512 & 76 & $<1$ \\
  \tabucline-
\end{tabu}}
\end{table}

Table~\ref{tab:results} shows the results of applying Theorem 3 on the production line model. For each control problem solved, the uncontrolled and controlled state-space size is provided. The control problems for synthesizing automaton-based supervisors $S_1,\ldots,S_5$ are tiny compared to monolithic synthesis, i.e., obtaining these supervisors can be done even manually. In future research, a full experimental analysis of potential computational effort reduction with respect to other synthesis algorithms can be performed. Inspecting the synthesized supervisors confirms the observation from Section~\ref{sect:casedescription} that no additional restrictions are imposed to ensure controllable and nonblocking behavior.

\section{\hlr{Case Study 2: Roadway tunnel}}\label{sect:tunnel}

In this section, the applicability of the proposed method \hlrr{is demonstrated} on an industrial large-scale system. For this demonstration, the case study of synthesizing a supervisory controller for the `Eerste Heinenoord Tunnel', a tunnel located south of Rotterdam, the Netherlands, \hlrr{is used}. The \hlr{model of this system is described} in~\cite{moormann_design_2020}.

\subsection{Case description}

\begin{figure}
\center
\includegraphics[width=\linewidth]{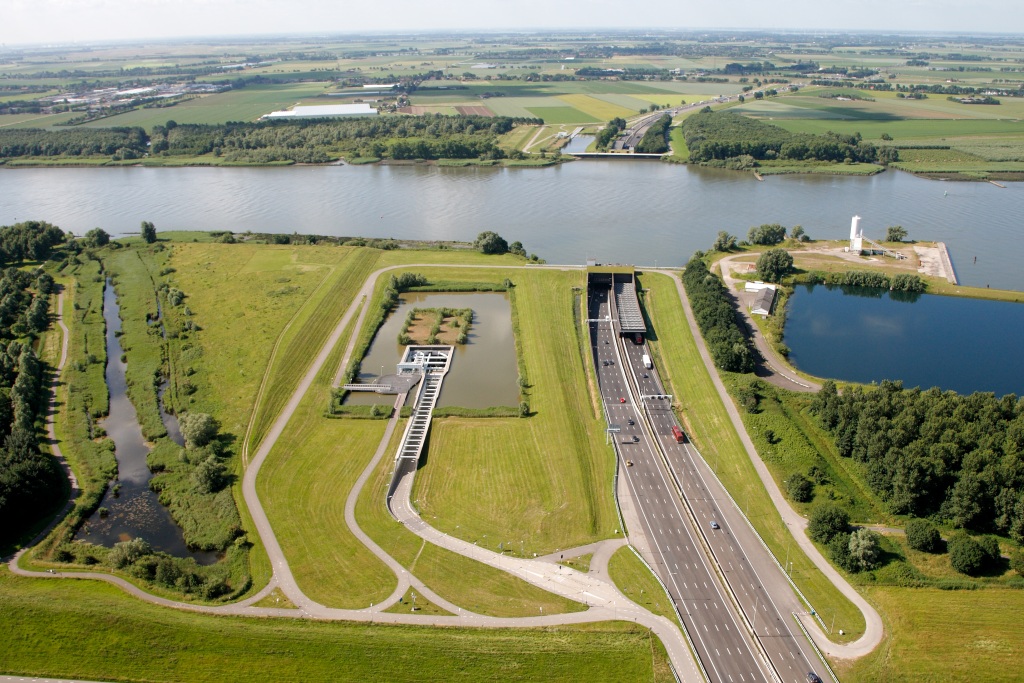}
\caption{The Eerste Heinenoord Tunnel (right) and the Tweede Heinenoord Tunnel (left). Image from https://beeldbank.rws.nl, Rijkswaterstaat.\label{fig:tunnel}}
\end{figure}

Nowadays, each tunnel is equipped with a supervisory controller that ensures correct cooperation between the tunnel subsystems, such as ventilation, lighting, boom barriers, and emergency detection sensors. For example, when an emergency is detected by several sensors, the supervisor has to automatically close off the tunnel for traffic. Figure~\ref{fig:tunnel} shows the `Eerste Heinenoord Tunnel' (EHT) on the right and the `Tweede Heinenoord Tunnel' (THT) on the left. The EHT is a two-tube roadway tunnel, which was initially opened in 1969. The THT, which was added in 1999, is only accessible for slow traffic such as cyclists and agricultural traffic. Rijkswaterstaat, the executive body of the Dutch ministry of Infrastructure and Water Management, is currently in the preparation and planning phase of renovating the EHT. In the renovation project, both the physical tunnel components and the tunnel supervisory controller are being renewed.

The model of the EHT in~\cite{moormann_design_2020} contains 540 plant models and 1668 requirement models, which can be accessed at a GitHub repository\footnote{\url{https://github.com/magoorden/NonblockingModularSupervisors}}. This large number of component models results in the uncontrolled state-space size of $1.87\cdot 10^{226}$, for which a monolithic supervisor can no longer be calculated by the CIF tooling~\citep{beek_cif_2014}.

\subsection{Results}
\begin{figure}
\center
\includegraphics[width=\linewidth]{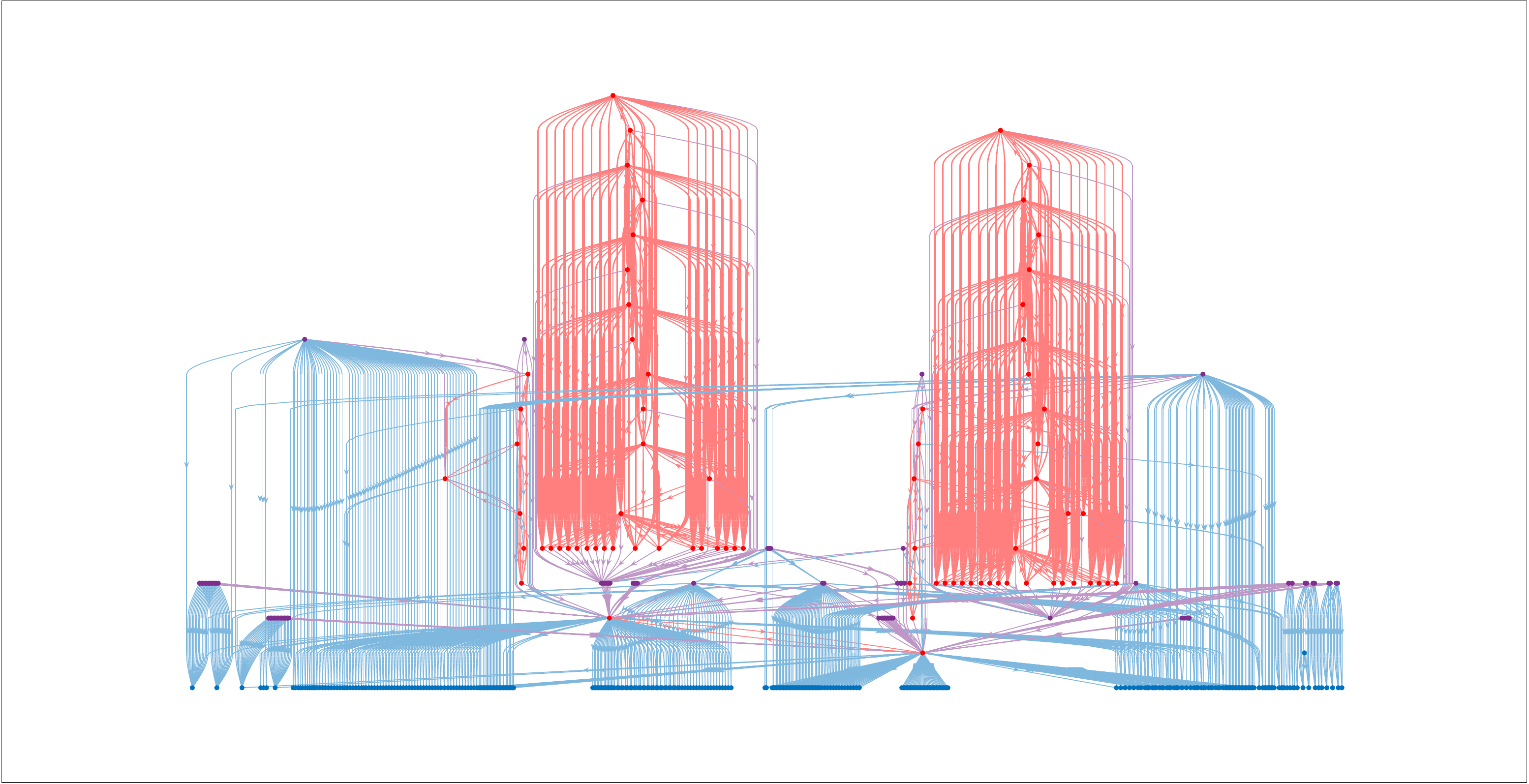}
\caption{The dependency graph of the EHT. Red \hlr{color} indicates the five strongly connected components, purple \hlr{color} the nodes and edges added in the extended strongly connected components, and blue \hlr{color} the nodes and edges that can be omitted from synthesis according to Theorem~\ref{thm:cyclicgraphs}.\label{fig:tunneldg}}
\end{figure}

The model of the EHT satisfies \textbf{RCNMS}, but it does not satisfy \textbf{CNMS}. Therefore, Theorem~\ref{thm:nosynthesis} does not apply. Figure~\ref{fig:tunneldg} shows the dependency graph of this model. Again, extended cycles are indicated in red in the figure. Since the dependency graph is cyclic, Theorem~\ref{thm:nosynthesisrelaxed} does not apply too. Therefore, Theorem~\ref{thm:cyclicgraphs} is used to reduce the synthesis problem.

With the help of Theorem~\ref{thm:cyclicgraphs}, instead of using the complete model as input for synthesis, the model can be significantly reduced. The dependency graph of the EHT model contains five strongly connected components, which transforms into one large subgraph of the five extended sets of vertices. Now, according to Theorem~\ref{thm:cyclicgraphs}, all blue vertices (and edges) can be removed before synthesis is started on the control problem represented by the red \hlr{and purple} edges and vertices.

\begin{table}
\renewcommand{\arraystretch}{1.5}
\vspace{0.5em}
\caption{Results of supervisory control synthesis for the EHT. Monolithic synthesis has been used.}
\label{tab:resultsEHT}
\centering
{\tabulinesep=1.2mm
\begin{tabu}{|X[1.5,l,m]|X[.9,r,m]|X[.9,r,m]|}
  \tabucline-
  \rowfont[l]{}
  Model & \raggedright  Original control problem & Reduced control problem \\ \tabucline[1.5pt]-
  Number of plant models & $492$ & $157$ \\ \tabucline-
  Number of requirement models & $1668$ & $1312$ \\ \tabucline-
  Uncontrolled state-space size & $1.87\cdot 10^{226}$ & $1.48\cdot 10^{87}$ \\ \tabucline-
  Controlled state-space size & - & $2.55\cdot 10^{81}$ \\ \tabucline-
  Synthesis duration [s] & - & 19.4  \\
  \tabucline-
\end{tabu}}
\end{table}

Table~\ref{tab:resultsEHT} shows the results of the analysis of the EHT. In the most-refined product representation, the EHT model contains $492$ plant models and $1668$ requirements. Theorem~\ref{thm:cyclicgraphs} reduces the synthesis problem to only $157$ plant models and $1312$ requirement models. This is a reduction of $68\%$ of the plant models and $21\%$ of the requirement models. Now the reduced model can be used as input for any synthesis method, e.g. monolithic, modular, and compositional synthesis, to obtain a supervisor. Monolithic synthesis \hlrr{is applied} to verify whether a supervisor can be synthesized for the reduced control problem without running into memory issues. For the reduced control problem, a monolithic supervisor can be synthesized in $19.4$ seconds. This shows that reducing the control problem is beneficial for synthesis.

\hl{As a subsequent experiment, multilevel synthesis~\citep{komenda_control_2016,goorden_structuring_2020} and compositional synthesis~\citep{mohajerani_framework_2014} are applied on the original model of the EHT. For multilevel synthesis, the implementation in CIF~\citep{beek_cif_2014} \hlrr{is used}; for compositional synthesis, the implementation in Supremica~\citep{malik_supremicaefficient_2017} \hlrr{is used}. Multilevel synthesis is able to synthesize supervisors on average in 220 seconds\footnote{With clustering settings of $\alpha=2, \beta=5.0$, and $\mu=2.0$.}. This is without performing a nonconflicting check on the synthesized supervisors. Both the monolithic BDD-based nonconflicting check in CIF and the compositional nonconflicting check in Supremica run out of memory (4GB available). Compositional synthesis is not able to synthesize a supervisor, because it runs out of memory (4GB available). This experiment shows that it is currently sometimes necessary to reduce the control problem before performing state-of-the-art synthesis algorithms on models of large-scale applications.}

\section{Conclusion}\label{sect:conclusion}

\hlr{This paper presents contributions} to determine, \hlr{based on model properties}, whether synthesis is unnecessary for a given set of \hlr{modular} plant models and requirement models, \hlr{building upon preliminary results presented in the conference paper of~\cite{goorden_no_2019}}. \hlll{These contributions result in the following \hllll{effective} \hlr{three-step} method for synthesizing supervisors. First, it is checked whether a control problem satisfies the \CNMS properties. If it does, then the synthesis step is altogether unnecessary: the plant and the requirement models already form a controllable, nonblocking, and maximally permissive supervisor. If not, the second step is to check whether the control problem satisfies the relaxed \RCNMS\ properties and to construct its dependency graph, \hlr{where vertices relate to the plant models and the edges to the requirement models}. If the dependency graph is acyclic, then the synthesis step is still unnecessary. If it has cycles, the third step is to reduce the original control problem to a collection of smaller partial control problems, one for each strongly connected component in the dependency graph.} \hlr{This results in modular supervisors which control the plant together.}

Two industrial cases studies demonstrate \hl{that} the method presented in this paper \hl{generates useful results in practice} \hlll{by significantly reducing the synthesis effort}. \hl{The tunnel case study even shows that model reduction is necessary, because state-of-the-art synthesis tools are not able to \hlll{provide} supervisors for the original model.}

The infrastructural systems encountered in the project with Rijkswaterstaat, like waterway locks~\citep{reijnen_supervisory_2017}, movable bridges~\citep{reijnen_supervisory_2019}, and tunnels~\citep{moormann_design_2020}, satisfy \textbf{RCNMS}. This is a motivation to further investigate the applicability of the proposed model properties and analysis method to systems from other domains, like manufacturing and automotive systems.

Future work also includes the identification of special cases to be able to conclude that the synthesis step is unnecessary for some of the partial control problems identified by the strongly connected components. Monolithic supervisors of the partial control problems of the production line case still indicate that the synthesis step is unnecessary, but it is yet unclear how this conclusion could be obtained without \hlr{having performed} synthesis. \hlr{Another direction is to investigate the applicability of the presented method if the supervisors are not required to be maximally permissive, i.e., if the goal is to synthesize controllable, and nonblocking supervisors that achieve something but not necessarily everything possible.}

\section*{Acknowledgment}
The authors thank Ferdie Reijnen for providing the models of the Festo production line and Lars Moormann for providing the models of the Eerste Heijnenoord Tunnel. Advice given by prof.em. Jan H. van Schuppen has been a great help in preparing the manuscript of this paper. The authors thank Han Vogel and Maria Angenent from Rijkswaterstaat for their valuable feedback and support.

This work is supported by Rijkswaterstaat, part of the Ministry of Infrastructure and Water Management of the Government of the Netherlands, and by the Swedish Science Foundation,
Vetenskapsr{\aa}det.


\bibliography{ref}

\appendix
\section{\hll{Proof of Theorem~\ref{thm:nosynthesis}}}\label{app:proofnosynthesis}
The proof of Theorem~\ref{thm:nosynthesis} as presented in this section originates from the conference proceedings~\citep{goorden_no_2019}. In order to prove that a control problem satisfying \textbf{CNMS} does not require synthesis (Theorem~\ref{thm:nosynthesis}), the following five lemmas \hlrr{are first proved}.

The first two lemmas show that when a plant model is provided as a product system and each individual automaton is trim or strongly connected, then the synchronous composition of these automata is also trim or strongly connected, respectively.

\begin{lemma}\label{lemma:istrim}
Let $\mathcal{P} = \{P_1,\ldots,P_m\}$ be a product system where each individual $P_i\in \mathcal{P}$ is trim. Then $P_1\parallel \ldots \parallel P_m$ is trim.
\end{lemma}
\begin{proof}
Denote $P=P_1\parallel \ldots \parallel P_n$, $P=(Q,\Sigma,\delta,q_0,Q_m)$, and $P_i=(Q_i,\Sigma_i,\delta_i,q_{0,i},Q_{m,i})$. \hlrr{It is shown} that $P$ is reachable and coreachable.

Firstly, assume that $q = (q_1,\ldots,q_n)$ is a state in $P$. As each individual $P_i$ is trim, it follows that there exists a string $s_i\in\Sigma_i^*$ such that $\delta_i(q_{0,i},s_i) = q_i$. From the definition of synchronous composition and the fact that $\mathcal{P}$ is a product system, it follows that $\delta((r_1,\ldots,q_{0,i},\ldots,r_m),s_i)= (r_1,\ldots,q_i,\ldots,r_m)$ for any state $r_j\in Q_j, j\neq i$. Therefore, it holds that $\delta((q_{0,1},\ldots,q_{0,n}),s_1s_2\ldots s_n) = q$ in $P$. As the state $q$ is chosen arbitrarily, it follows that $P$ is reachable.

Secondly, assume again that $q = (q_1,\ldots,q_n)$ is a state in $P$. As each individual $P_i$ is trim, it follows that there exists a string $s_i\in\Sigma_i^*$ such that $\delta_i(q_i,s_i) = q_{i,k} \in Q_{m,i}$. From the definition of synchronous composition and the fact that $\mathcal{P}$ is a product system, it follows that $\delta((r_1,\ldots,q_i,\ldots,r_m),s_i)= (r_1,\ldots,q_{i,k},\ldots,r_m)$ for any state $r_j\in Q_j, j\neq i$. Therefore, it holds that $\delta(q,s_1s_2\ldots s_n) \in Q_m$ in $P$. As state $q$ is chosen arbitrarily, it follows that $P$ is coreachable.
\end{proof}

\begin{lemma}\label{lemma:isstronglyconnected}
Let $\mathcal{P} = \{P_1,\ldots,P_m\}$ be a product system where each individual $P_i\in \mathcal{P}$ is a strongly connected automaton. Then $P_1\parallel \ldots \parallel P_m$ is a strongly connected automaton.
\end{lemma}
\begin{proof}
Denote $P=P_1\parallel \ldots \parallel P_n$, $P=(Q,\Sigma,\delta,q_0,Q_m)$, and $P_i=(Q_i,\Sigma_i,\delta_i,q_{0,i},Q_{m,i})$. \hlrr{It is shown} that for any two states $x = (x_1,\ldots,x_m)\in Q ,y = (y_1,\ldots,y_m)\in Q$ there exists a string $s\in\Sigma^*$ such that $\delta(x,s) = y$.

As each individual $P_i$ is strongly connected, it follows that there exists a string $s_i\in\Sigma_i^*$ such that $\delta_i(x_i,s_i) = y_i$. From the definition of synchronous composition and the fact that $\mathcal{P}$ is a product system, it follows that $\delta((r_1,\ldots,x_i,\ldots,\allowbreak r_m), s_i)= (r_1,\ldots,y_i,\ldots,r_m)$ for any state $r_j\in Q_j, j\neq i$. Therefore, it holds that $\delta(x,s_1s_2\ldots s_n) = y$ in $P$. As states $x$ and $y$ are chosen arbitrarily, it follows that $P$ is a strongly connected automaton.
\end{proof}

The following lemma expresses that when a control problem with a single requirement satisfies \textbf{CNMS}, then always eventually a state \hlrr{can be reached} such that the condition of this requirement evaluates to true, thus enabling the guarded event.

\begin{lemma}\label{lemma:makeconditiontrue}
Let $(\mathcal{P},\{R\})$ be a control problem with a single requirement satisfying \textbf{CNMS}. Denote $R=e\textbf{ needs } C$. Then, from any state $q$, there exists a string $s\in\Sigma^*$ such that a state $r$ is reached and $C(r) = \textbf{T}$.
\end{lemma}
\begin{proof}
As $\mathcal{P}$ is a product system (Property 1), there is only a single plant component $P_k$ such that $e\in\Sigma_k$. From the combination of Properties 3.b, 3.d, and 3.g, it follows that plant component $P_k$ is not used in condition $C$, as it has to be an actuator model. Therefore, the state of $P_k$ does not matter.

Furthermore, observe that $\mathcal{P}\setminus\{P_k\}\neq\emptyset$ and $\parallel (\mathcal{P}\setminus\{P_k\})=\parallel (\mathcal{P}\setminus\{P_k\})\parallel R$. From Property 2 and Lemma~\ref{lemma:isstronglyconnected} it follows that $\parallel (\mathcal{P}\setminus\{P_k\})$ is a strongly connected automaton, thus $\parallel (\mathcal{P}\setminus\{P_k\}) \parallel R$ is also a strongly connected automaton. Therefore, if there exists a state $r$ that satisfies $C$, i.e., $C(r)=\textbf{T}$, then there also exists a string $s\in\Sigma^*$ such that $\delta(q,s) = r$. So it remains to be proven that such a state $r$ exists.

As $C$ is in disjunctive normal form (Property 3.d), it follows that if $r$ satisfies $C$, it satisfies one of the conjunctions. From Properties 3.e and 3.g \hlrr{it is known} that there is at most one reference to each $P_i \in \mathcal{P}\setminus\{P_k\}$ in each conjunction. If there is no reference to $P_i$, then all states of this automaton satisfy this conjunction. If $P_i$ is mentioned in this conjunction, then, from Properties 3.d and 3.f, there exists at least one state $q_i\in Q_i$ that satisfies this conjunction. Thus there exists a state $r$ such that $C$ is satisfied.
\end{proof}

Now the following two lemmas \hlrr{are proven}: the first one shows that under the given conditions, synthesis \hlrr{is not needed to be performed} locally, and the second one shows that under the given conditions the supervisors are globally nonblocking. In the rest of this section, the notation $\sup\mathit{CN}(P,R)$ is the function that constructs the controllable, nonblocking, and maximally permissive supervisor given plant $P$ and requirement $R$.

\begin{lemma}\label{thm:indsupervisorisplantreq}
Let $(\mathcal{P},\mathcal{R})$ be a control problem satisfying \textbf{CNMS}. For each $R_j\in\mathcal{R}$, $P\parallel R_j$ is a controllable, nonblocking, and maximally permissive supervisor for plant $P=\parallel \mathcal{P}$ and requirement $R_j$.
\end{lemma}
\begin{proof}
In the case that $\mathcal{R}=\emptyset$, no supervisor is synthesized. It follows from Properties 1 and 2 and Lemma~\ref{lemma:istrim} that $P$ is trim, so there is indeed no need for a supervisor. In the remainder of the proof \hlrr{it is assumed} that $\mathcal{R}\neq\emptyset$.

For each individual supervisor $P\parallel R_j$ \hlrr{it is shown below} that $P\parallel R_j$ is controllable with respect to plant $P$ and that $P\parallel R_j$ is nonblocking. The fact that $P\parallel R_j$ is controllable follows directly from Property 3.b. It remains to be proven that $P\parallel R_j$ is nonblocking. From Property 3.a \hlrr{it follows that} an event $e_j=\mathit{event}(R_j)$ \hlrr{is} associated with this requirement $R_j$. As $\mathcal{P}$ is a product system (Property 1), there is only a single plant component $P_k$ such that $e_j\in\Sigma_k$. Now the set of plant component models \hlrr{is partitioned} into $\{P_k\}$, $P_{\mathit{sm}} = \{ P_i\in \mathcal{P}\ |\ P_i \text{ is a sensor model}\}$, and $P_o = \mathcal{P}\setminus(\{P_k\}\cup P_{\mathit{sm}})$. Observe that the behavior of the plant components in $P_{\mathit{sm}}$ and $P_o$ is not restricted by requirement $R_j$, so Lemmas~\ref{lemma:istrim} and~\ref{lemma:isstronglyconnected} apply to the sets $P_{\mathit{sm}}$, $P_o$, and $P_{\mathit{sm}}\cup P_o$, i.e, $P_{\mathit{sm}}\parallel R_j$, $P_o\parallel R_j$, and $(P_{\mathit{sm}}\cup P_o)\parallel R_j$ are all trim and strongly connected automata.

To show that $P\parallel R_j$ is nonblocking, \hlrr{it is shown next} that for each reachable state $q$ there exists a string $s\in\Sigma^*$ such that a marked state $q_m\in Q_m$ can be reached. Consider automaton $P_k$ with current state $q_k$. As automaton $P_k$ is trim (Property 2), there exists a path labeled with string $s_k\in\Sigma_k^*$ by which a state $q_{m,k}\in Q_{m,k}$ can be reached from state $q_k$. \hlrr{It is shown} that this path is still possible under the influence of requirement $R_j$, i.e., it is still a path in $P_k\parallel R_j$. Consider two cases for this path.
\begin{itemize}
    \item If $s_k$ does not contain event $e_j$, then the path labeled with $s_k$ is trivially possible in $P_k\parallel R_j$.
    \item If $s_k$ contains event $e_j$, then requirement $R_j$ may remove event $e_j$ from the enabled event set and prevents $P_k\parallel R_j$ from reaching a marked state. For each transition labeled with event $e_j$, Lemma~\ref{lemma:makeconditiontrue} \hlrr{expresses} that there exists a path in $P$ reaching a state $r$ such that $C(r) = \textbf{T}$. Therefore, there always exists a path in $P$ such that $e_j$ is enabled. Thus, the path labeled with $s_k$ is still possible in $P_k\parallel R_j$.
\end{itemize}
Combining the above observation for $s_k$ and the fact that $(P_{\mathit{sm}}\cup P_o)\parallel R_j$ is trim, \hlrr{it follows} that a string $s$ exists by which a marked state $q_m$ is reached from state $q$. As $q$ is arbitrarily chosen, it follows that $P\parallel R_j$ is nonblocking.
\end{proof}

\begin{lemma}\label{thm:modulararenonconflicting}
Let $(\mathcal{P},\mathcal{R})$ be a control problem satisfying \textbf{CNMS}. Construct the set of modular supervisors $\mathcal{S} = \{S_1,\ldots,S_n\}$ such that each supervisor $S_j = \sup\mathit{CN}(P,R_j)$ is the controllable, nonblocking, and maximally permissive supervisor for plant $P=P_1\parallel\ldots\parallel P_m$ and requirement $R_j\in R$. Then $\mathcal{S}$ is nonconflicting.
\end{lemma}
\begin{proof}
For $\mathcal{S}$ to be nonconflicting, it should hold that $S_1\parallel \ldots\parallel S_n$ is nonblocking. From Lemma~\ref{thm:indsupervisorisplantreq} it follows that each $S_j = P\parallel R_j$. Therefore, $S_1\parallel \ldots\parallel S_n = (P\parallel R_1)\parallel\ldots\parallel (P\parallel R_n) = P\parallel R_1\parallel \ldots\parallel R_n$. Partition the set of plant models $\mathcal{P}$ into the set of sensor models $P_{\mathit{sm}} = \{ P_i\in \mathcal{P}\ |\ P_i \text{ is a sensor model}\}$, the set of restricted models $P_r = \{P_i\in \mathcal{P}\ |\ \exists R_j\in \mathcal{R} \text{ s.t. } \mathit{event}(R_j) \in \Sigma_i\}$, and the other plant models $P_o = \mathcal{P}\setminus(P_{\mathit{sm}}\cup P_r)$.

Clearly, no plant model in $P_o$ is affected by the requirements, so Lemmas~\ref{lemma:istrim} and~\ref{lemma:isstronglyconnected} apply, i.e., $P_o\parallel \mathcal{R}$ is a trim and strongly connected automaton. Furthermore, from Property 3.b and the definition of a sensor model it follows that also no plant model in $P_{\mathit{sm}}$ is affected by the requirements, thus by Lemmas~\ref{lemma:istrim} and~\ref{lemma:isstronglyconnected} it follows that $P_{\mathit{sm}}\parallel \mathcal{R}$ is a trim and strongly connected automaton. Again using Lemmas~\ref{lemma:istrim} and~\ref{lemma:isstronglyconnected} yields that $P_o\parallel P_{\mathit{sm}}\parallel \mathcal{R}$ is a trim and strongly connected automaton.

For $P_o\parallel P_{\mathit{sm}} \parallel P_r \parallel \mathcal{R}$ to be nonblocking, it should hold that from every reachable state $q\in Q$ there exists a string $s\in\Sigma^*$ such that $\delta(q,s)\in Q_m$. As $P_r$ is trim (Lemma~\ref{lemma:istrim}) it follows that there exists a string $s_r\in\Sigma_r^*$ such that $\delta(q_r,s_r)\in Q_m$ in $P_r$. For $\delta(q_r,s_r)\in Q_m$ in $P_r\parallel \mathcal{R}$ to exist, each event in $s_r$ should be enabled along its path. There are two cases for each event $\sigma$ in string $s_r$ to consider following Definition~\ref{def:syncompreq} of synchronous composition with a state-event requirement.
\begin{itemize}
    \item If there does not exist a requirement $R_j\in \mathcal{R}$ such that $\mathit{event}(R_j) = \sigma$, then $\sigma$ is enabled.
    \item If there does exist a requirement $R_j\in \mathcal{R}$ such that $\mathit{event}(R_j) = \sigma$, then $R_j$ is also the only requirement in $\mathcal{R}$ such that $\mathit{event}(R_j) = \sigma$ (Property 3.c). As the condition $C_j = \mathit{cond}(R_j)$ only depends on plant components from $P_{\mathit{sm}}$ and not plant components from $P_r$ or $P_o$ (Property 3.g), it follows from Lemma~\ref{thm:indsupervisorisplantreq} that there exists a string in $P_{\mathit{sm}}$ such that the reached state $r$ satisfies $C_j$. No transition in plant components from $P_r$ and $P_o$ are needed as all states from these plant components are irrelevant in satisfying the condition $C_j$. Therefore, there exists a path in $P$ such that $\sigma$ is enabled.
\end{itemize}
From the above observation, \hlrr{it can be concluded} that always a string (including the empty string) such that $\sigma$ is enabled \hlrr{can be found}. As $\sigma$ is chosen arbitrarily along the path in $P_r$ labeled with $s_r$, it follows that $\delta(q_r,s_r)\in Q_{m,r}$. Finally, combining this with the fact that $q_r$ is chosen arbitrarily and that $P_o\parallel P_{\mathit{sm}}\parallel \mathcal{R}$ is trim, it follows that $P_o\parallel P_{\mathit{sm}}\parallel P_r\parallel \mathcal{R}$ is nonblocking.
\end{proof}

Now Theorem~\ref{thm:nosynthesis} \hlrr{is proven}.

\begin{proof}[Proof of Theorem~\ref{thm:nosynthesis}]
From Lemmas~\ref{thm:indsupervisorisplantreq} and~\ref{thm:modulararenonconflicting} it follows that a set of supervisors $\mathcal{S}=\{S_1,\ldots,S_n\}$ \hlrr{can be constructed} such that $S_j = \sup\mathit{CN}(P,R_j) = P\parallel R_j$ and $\mathcal{S}$ is nonconflicting. The antecedent follows directly from combining these last two facts.
\end{proof}

\section{\hll{Proof of Theorem~\ref{thm:nosynthesisrelaxed}}}\label{app:proofnosynthesisrelaxed}
\begin{figure}
	\begin{center}
		\begin{tikzpicture}[every node/.style={font=\sffamily\small, color=black}, graph node/.style={circle,draw=lbl,font=\small,fill=lbl, inner sep=0pt, minimum size=3pt}, node distance=1cm, 
			->-/.style={thin, color=lbl, decoration={
					markings,
					mark=at position 0.6 with {\arrow{stealth}}}, postaction={decorate}}
			]
			\node[graph node, label={above:$P_1$}] (v1) {};
			\node (v) [below = of v1] {};
			\node[graph node, label={left:$P_2$}]  (v2) [left  = of v.center] {};
			\node[graph node, label={right:$P_3$}] (v3) [right = of v.center] {};
			\node[graph node, label={right:$P_5$}] (v4) [below = of v] {};
			\node[graph node, label={right:$P_5'$}] (v4t) [right = of v4.center] {};
			\node[graph node, label={left:$P_4$}] (v5) [left  = of v4.center] {};
			
			\draw[->-] (v1) -- node [left]  {$e_1$} (v2);
			\draw[->-] (v1) -- node [right] {$e_2$} (v3);
			\draw[->-] (v1) -- node [right] {$e_3$} (v4);
			\draw[->-] (v3) -- node [right] {$e_5$} (v4t);
			\draw[->-] (v2) -- node [left]  {$e_4$} (v5);
		\end{tikzpicture}
	\end{center}
	\caption{The forest of dependency graph $G_{cp}$ from Figure~\ref{fig:exampleGraph2}.}
	\label{fig:exampleForest}
\end{figure}
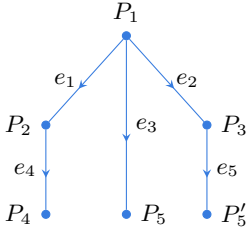

Before Theorem~\ref{thm:nosynthesisrelaxed} \hlrr{is proven}, the following lemma is introduced which transforms an acyclic dependency graph into a forest of trees. A \emph{tree} is an acyclic directed graph where each vertex has at most one incoming edge, i.e., for each vertex $v$ there is at most one edge $e$ such that $\mbox{ter}(e)=v$. A forest is a set of trees. A forest can be constructed from an acyclic directed graph recursively. Assume that a subgraph $T$ having vertex $v$ as root node is already a tree. Then for each incoming edge into $v$ subgraph $T$ is duplicated and set to the terminating vertex of that edge. Figure~\ref{fig:exampleForest} shows the forest with a single tree of the dependency graph as shown in Figure~\ref{fig:exampleGraph2}. As vertex $P_5$ has two incoming edges, the directed graph $G_{cp}$ is not a tree. By duplicating vertex $P_5$, the tree in Figure~\ref{fig:exampleForest} is constructed.

From a dependency (sub)graph, the control problem it represents can be reconstructed as follows. The control problem $(\mathcal{P}, \mathcal{R}')$ represented by a dependency graph $(\mathcal{P}, E)$ is the one where $\mathcal{R}' = \{R\in \mathcal{R}\ |\ \exists e \in E \mbox{ s.t. } \mathit{event}(R) \in \Sigma_{\mbox{init}(e)} \wedge\mbox{ ter}(e) \in \mathit{cond}(R) \}$.

\begin{lemma}\label{lem:forest}
Let $G_{\mathit{CP}} = (\mathcal{P},E)$ be an acyclic dependency graph of control problem $\mathit{CP} = (\mathcal{P},\mathcal{R})$ satisfying \textbf{RCNMS}, and let $F$ be the forest constructed from $G_{\mathit{CP}}$. Then $S$ is a controllable, nonblocking, and maximally permissive supervisor of $\mathit{CP}$ if and only if $S$ is a controllable, nonblocking, and maximally permissive supervisor of the control problem \hlr{represented by $F$}.
\end{lemma}
\begin{proof}
In the construction of the forest $F$ from $G_{\mathit{CP}}$, subgraphs are duplicated. Duplicating plants and requirements results in the same controllable, nonblocking, and maximally permissive supervisor, i.e., $S$ is a controllable, nonblocking, and maximally permissive supervisor for $(\mathcal{P}'\parallel \mathcal{P}') \parallel (\mathcal{R}'\parallel \mathcal{R}')$ if and only if $S$ is a controllable, nonblocking, and maximally permissive supervisor for $\mathcal{P}' \parallel \mathcal{R}'$, where $\mathcal{P}'\subseteq \mathcal{P}$ and $\mathcal{R}'\subseteq \mathcal{R}$ are sets of plant models and requirement models, respectively. As forest $F$ is constructed recursively in this manner, the result holds for the complete forest.
\end{proof}

\hll{The proof of Theorem~\ref{thm:nosynthesisrelaxed} follows now.}
\begin{proof}[Proof of Theorem~\ref{thm:nosynthesisrelaxed}]
For $\mathit{CP} = (\mathcal{P}, \mathcal{R})$, $\mathit{CP}' = (\mathcal{P}', \mathcal{R}')$ is a partial control problem of $\mathit{CP}$, denoted by $\mathit{CP}'\preceq \mathit{CP}$, if $\mathcal{P}'\subseteq \mathcal{P}$ and $\mathcal{R}'\subseteq \mathcal{R}$. From Lemma~\ref{lem:forest} it follows that the forest $F$ constructed from $G_{\mathit{cp}}$ can be analyzed instead of $G_{\mathit{cp}}$ directly. Therefore, \hlrr{it is shown next} that no synthesis is needed if (each tree in) the forest is acyclic by induction on the depth of each tree in forest $F$.

\textbf{Base case} Let subgraph $(\mathcal{P}', \emptyset)\subseteq F$ with $\mathcal{P}'\subseteq \mathcal{P}$ be a tree of depth zero, i.e., it only contains leaf nodes. Then the partial control problem $(\mathcal{P}', \emptyset)$ represented by this subgraph is trivially controllable and nonblocking, and $\mathcal{P}'$ is strongly connected.

\textbf{Induction hypothesis} Assume the set of trees $\{T_1,\ldots,T_k\}$ each with depth at most $n$ such that for each tree $(\mathcal{P}^i,E^i), i\in [1,k]$ the partial control problem $(\mathcal{P}^i, \mathcal{R}^i)$ represented by this subgraph is controllable and nonblocking, and $\mathcal{P}^i\parallel \mathcal{R}^i$ is strongly connected.

\textbf{Inductive step} Denote $\mathcal{P}' = \mathcal{P}^1\cup\ldots\cup \mathcal{P}^k$ the set of all vertices and $E'= E^1\cup\ldots\cup E^k$ the set of all edges of the trees with depth at most $n$, and the control problem $(\mathcal{P}',\mathcal{R}')$ represented by subgraph $(\mathcal{P}',E')$. Let $P \in \mathcal{P} \setminus \mathcal{P}'$ be a vertex not yet in any tree of depth at most $n$ such that for all edges $e\in E$ with $\hlr{\mbox{init}(e) = P}$, \hlr{which} \hlrr{is assigned} to $E_P$, it holds that $\mbox{ter}(e)\in \mathcal{P}'$. \hlr{With other words, $P$ is selected if each of its outgoing edges enter a tree with depth at most $n$. Note that each edge in $E_p$ has a different terminal vertex, because $F$ is a forest.} Let $\mathcal{R}=\{R\in \mathcal{R}\ |\ \mathit{event}(R) \in \Sigma_P\}$ contain all requirements restricting the behavior of $P$. \hlrr{It is shown below} that the partial control problem represented by tree $(\mathcal{P}'\cup \{P\},E'\cup E_P)$ of depth at most $n+1$ is controllable and nonblocking, and strongly connected.

From the induction hypothesis it follows that the control problem represented by subgraph $(\mathcal{P}',E')$ is strongly connected, i.e., $\mathcal{P}'\parallel \mathcal{R}'$ is strongly connected. Therefore, similarly to Lemma 3 of~\cite{goorden_no_2019}, for each requirement $R\in \mathcal{R}$ there exists a string such that state $r$ of $\mathcal{P}'$ is reached that satisfies the condition $\mathit{cond}(R)$, thus enabling controllable event $\mathit{event}(R)$. Analogously to the proof of Lemma 4 of~\cite{goorden_no_2019}, it holds that a string can be constructed in $\mathcal{P}'$ such that for each path in plant $P$ all controllable events are enabled. Therefore, the partial control problem $(\mathcal{P}'\cup\{P\},\mathcal{R}'\cup\mathcal{R})$ represented by subgraph $(\mathcal{P}'\cup \{P\},E'\cup E_P)$ of depth at most $n+1$ is controllable and nonblocking, and for each $P_i\in \mathcal{P}'\cup\{P\}$ all states can be reached from each state. This concludes the inductive step.
\end{proof}
\section{\hll{Proof of Theorem~\ref{thm:cyclicgraphs}}}\label{app:proofcyclicgraphs}
\begin{proof}[Proof of Theorem~\ref{thm:cyclicgraphs}]
First, let $\mathcal{P}'$ be the set of all vertices \emph{not} contained in $\mathbb{W}$, i.e., $\mathcal{P}' = \mathcal{P} \setminus (\bigcup_{W\in\mathbb{W}} \bigcup_{V\in W} V)$. From the definition of $V$ it follows that for each vertex $P\in \mathcal{P}'$ there does not exist a path to a cycle. Therefore, the subgraph $(\mathcal{P}',E')$ with $E' = \{e\in E\ |\ \mbox{init}(e)\in \mathcal{P}'\}$ is acyclic. From Theorem~\ref{thm:nosynthesisrelaxed} it directly follows that control problem $(\mathcal{P}', \mathcal{R}')$ represented by this acyclic graph is already controllable and nonblocking, i.e., synthesis can be skipped for this part. Furthermore, from the proof of that theorem it follows that $\mathcal{P}'\parallel \mathcal{R}'$ is also strongly connected.

Now, consider $W\in\mathbb{W}$. The simplified partial control problem $(P_W,\tilde{R}_W)$ represented by $P_W=\bigcup_{V\in W} V$ may contain requirements where the condition is simplified by replacing some state references $P.q$ by $\textbf{T}$. As $\mathbb{W}$ is the quotient set of $\mathbb{V}$ by $\sim$, it follows from the definition of $\sim$ that each plant $P$ of those replaced state references is from $\mathcal{P}'$. As \hlrr{it is} already shown in the previous paragraph that $\mathcal{P}'\parallel \mathcal{R}'$ is strongly connected, it is always possible to reach state $P.q$, which justifies the replacement of this state reference by $\textbf{T}$. Therefore, if $S_W$ is a controllable, nonblocking, and maximally permissive supervisor of the simplified partial control problem $(P_W,\tilde{R}_W)$, then $\mathcal{P}'\parallel \mathcal{R}'\parallel S_W$ is a controllable, nonblocking, and maximally permissive supervisor for the partial control problem $(\mathcal{P}'\cup P_W, \mathcal{R}'\cup R_W)$, with $R_W$ the non-simplified requirements of $\tilde{R}_W$.

From the definition of $\mathbb{W}=\mathbb{V}/\sim$, it follows that no vertex is shared between two distinct $W_1,W_2\in\mathbb{W}, W_1\neq W_2$, i.e., $(\bigcup_{V\in W_1} V) \cap (\bigcup_{V\in W_2} V) = \emptyset$. Let $S_{W_1}$ and $S_{W_2}$ be the controllable, nonblocking, and maximally permissive supervisors for the simplified control problems represented by $W_1$ and $W_2$, respectively. Then the supervisors do not share events and it holds trivially that $S_{W_1}\parallel S_{W_2}$ is a controllable, nonblocking, and maximally permissive supervisor.

Finally, combining the above observations for each $W\in\mathbb{W}$ it follows that $\mathcal{P}'\parallel \mathcal{R}'\parallel (\parallel_{W\in\mathbb{W}}S_W) = \mathcal{P}\parallel \mathcal{R}\parallel (\parallel_{W\in\mathbb{W}}S_W)$ is a controllable, nonblocking, and maximally permissive supervisor.
\end{proof}
\end{document}